\newif\iflongversion
\longversiontrue

\documentclass{article}
\pdfoutput=1

\usepackage{amsmath}
\usepackage{amsthm, thmtools}
\usepackage{amsfonts}
\usepackage{mathrsfs}
\usepackage{amssymb}
\usepackage{bm}
\usepackage{mathtools}
\usepackage{physics}
\usepackage{stmaryrd}
\usepackage[dvipsnames]{xcolor}
\usepackage[pdfusetitle,bookmarksnumbered,colorlinks,linkcolor=MidnightBlue,citecolor=olive]{hyperref}
\usepackage[capitalise,nameinlink]{cleveref}
\usepackage[british]{babel}
\usepackage{textcomp}

\usepackage[a4paper, margin=3.2cm]{geometry}
\usepackage[utf8]{inputenc}
\DeclareFontSeriesDefault[rm]{bf}{b}
\newcommand{\defn}[1]{{\normalfont\textbf{#1}}}
\usepackage{packages/tikzit}

\tikzstyle{label}=[font={\footnotesize}, text height=1ex, text depth=0.15ex]
\tikzstyle{map}=[draw, shape=rectangle, inner sep=2pt, minimum height=4mm, fill=white, minimum width=5mm]
\tikzstyle{tiny map}=[draw, shape=rectangle, inner sep=1pt, minimum height=3mm, fill=white, minimum width=3mm, font={\footnotesize}]
\tikzstyle{small map}=[draw, shape=rectangle, inner sep=1pt, minimum height=4mm, fill=white, minimum width=4mm]
\tikzstyle{medium map}=[draw, shape=rectangle, inner sep=2pt, minimum height=4mm, fill=white, minimum width=9mm, tikzit fill=red]
\tikzstyle{large map}=[draw, shape=rectangle, inner sep=2pt, minimum height=4mm, fill=white, minimum width=18mm, tikzit fill=blue]
\tikzstyle{upground}=[circuit ee IEC, thick, ground, rotate=90, scale=1.5, tikzit fill=black, tikzit shape=rectangle]
\tikzstyle{downground}=[circuit ee IEC, thick, ground, rotate=-90, scale=1.5, tikzit fill=black, tikzit shape=rectangle]
\tikzstyle{downgroundnorm}=[circuit ee IEC, thick, ground, rotate=-90, scale=1.5, fill=white, tikzit shape=rectangle, tikzit fill=red]
\tikzstyle{object}=[text depth=0.15ex]
\tikzstyle{small label}=[label, font={\scriptsize}]

\tikzstyle{dir}=[->]
\tikzstyle{mapsto}=[{|->}]
\tikzstyle{dashdir}=[->, dashed]
\tikzstyle{merge}=[-, dashed, draw=lightgray]
\tikzstyle{map}=[draw=lightgray, {|->}]

\newcommand{\tikzfigpad}[1]{\quad\input{#1.tikz}\quad}
\usepackage{enumerate}
\usepackage{caption}
\usepackage{subcaption}
\usepackage{placeins} 

\usepackage[style=numeric-comp, 
    sortcites=true,
    useprefix=true,
    sorting=none,
    maxbibnames=9,
    doi=false,isbn=false,url=false,eprint=true]{biblatex}

\DeclareSortingNamekeyTemplate{
    \keypart{\namepart{family}}
    \keypart{\namepart{prefix}} 
    \keypart{\namepart{given}}
    \keypart{\namepart{suffix}}
}
\renewbibmacro{begentry}{\midsentence}  

\newbibmacro{string+doiurl}[1]{
  \iffieldundef{doi}{%
    \iffieldundef{url}{%
      #1
    }{%
      \href{\thefield{url}}{#1}%
    }%
  }{%
    \href{http://dx.doi.org/\thefield{doi}}{#1}%
  }%
}
\DeclareFieldFormat{title}{\usebibmacro{string+doiurl}{\mkbibemph{#1}}}
\DeclareFieldFormat[thesis]{title}{\usebibmacro{string+doiurl}{\mkbibemph{#1}}} 
\DeclareFieldFormat[article,misc]{title}{\usebibmacro{string+doiurl}{\mkbibquote{#1}}}
\DeclareFieldFormat{titlecase}{#1} 

\AtEveryBibitem{\clearfield{month}} 
\AtEveryBibitem{\clearfield{day}} 

\DeclareFieldFormat{postnote}{#1} 
\renewbibmacro{in:}{\ifentrytype{article}{}{\printtext{\bibstring{in}\intitlepunct}}} 
\newcommand{\labelsA}{A}
\newcommand{\labelsB}{B}
\newcommand{\lA}{a}
\newcommand{\lB}{b}
\newcommand{\cAg}{\bar{\labelsA}}
\newcommand{\cBg}{\bar{\labelsB}}
\newcommand{\powset}{\mathcal{P}}
\newcommand{\gC}{\mathfrak{c}}  
\newcommand{\gP}{\mathfrak{p}}
\newcommand{\ChG}{\ensuremath{G}}
\newcommand{\PaG}{\ensuremath{G^{-1}}}
\newcommand{\lap}{\la_{\cls{G}}}
\newcommand{\mup}{\mu_{\cls{G}}}
\newcommand{\card}[1]{|#1|}
\let\a\alpha
\let\b\beta
\let\la\lambda
\let\phi\varphi
\let\<\langle
\let\>\rangle

\let\bigjoin\bigvee

\let\bigmeet\bigwedge
\DeclareMathOperator{\id}{id}
\DeclareMathOperator{\im}{im}
\newcommand{\relG}{G}
\newcommand{\cls}[1]{\ensuremath{{B_{#1}}}} 
\newcommand{\latt}[1]{\ensuremath{{L_{#1}}}} 
\newcommand{\lattA}{\powset^c(\labelsA)}
\newcommand{\lattB}{\powset^c(\labelsB)}
\newcommand{\lattG}{\latt{G}}
\newcommand{\lattBop}{\powset^c(\labelsB)^\text{op}}
\newcommand{\incomp}{\mathrel{\#}} 
\newcommand{\cat}[1]{\mathsf{#1}}
\newcommand{\CircuitPreOrder}{\cat{Circ}_{AB}^{\preceq}}

\newcommand{\itmref}[1]{(\ref{#1})}

\newcommand{\cccG}{C_3}
\DeclareMathOperator{\End}{End}

\newcommand{\canonicalcircuit}{concept lattice}
\newcommand{\primitivecircuit}{basic circuit}

\theoremstyle{definition}
\declaretheorem[name=Definition,refname={Definition,Definitions},Refname={Definition,Definitions},qed=$\sslash$]{definition}
\declaretheorem[name=Example,refname={Example,Examples},Refname={Example,Examples},qed=$\sslash$,numberlike=definition]{example}
\declaretheorem[style=remark,name=Remark,numbered=no,qed=$\sslash$,numberlike=definition]{remark}

\theoremstyle{plain}
\declaretheorem[name=Theorem,refname={Theorem,Theorems},Refname={Theorem,Theorems},numberlike=definition]{theorem}
\declaretheorem[name=Proposition,refname={Proposition,Propositions},Refname={Proposition,Propositions},numberlike=definition]{proposition}

\declaretheorem[name=Corollary,refname={Corollary,Corollaries},Refname={Corollary,Corollaries},numberlike=definition]{corollary}

\title{\bfseries An order-theoretic circuit syntax and characterisation of the concept lattice}

\author{Tein van der Lugt \\[.7em] \normalsize Department of Computer Science, University of Oxford}
\date{\today}

\begin{document}

    \maketitle

    \begin{abstract}
        We take an order-theoretic approach to circuit (string diagram) syntax, treating a circuit as a partial order with additional input-output structure.
        We define morphisms between circuits and prove a factorisation theorem showing that these can, in the finite case, be regarded as formalising a notion of syntactical circuit rewrites, with quotient maps in particular corresponding to gate composition.
        We then consider the \emph{connectivity} of a circuit, expressed as a binary relation between its inputs and outputs, and characterise the \emph{concept lattice} from formal concept analysis as the unique smallest circuit that admits morphisms from all other circuits with the same connectivity.
        This has significance for quantum causality, particularly to the study of causal decompositions of unitary transformations.
        We close by constructing the circuit characterised by the dual statement.
    \end{abstract}

    \section{Introduction}\label{sec:intro}
This paper is about the syntax of \emph{circuits}, or \emph{string diagrams}, which describe parallel and sequential compositions of processes in symmetric monoidal categories (SMCs).
A particular motivating application for this work, however, lies in causal inference and quantum foundations, where classical or quantum circuits may act as models of causal processes.
In this context, the syntax of a circuit---i.e.\ its `shape'---alone implies relevant information about the process (morphism of the SMC) it implements: for example, the absence of a path between an input and an output of the circuit prohibits the flow of causal influence between those systems.
Generally, one may wonder about the full class of morphisms in a specified SMC---be it classical probability theory, quantum theory, or any other---that can be expressed by a circuit with a given shape, a type of problem sometimes referred to as \emph{causal compatibility}.%
\footnote{See e.g.~\cite{WSF19}. In quantum and classical causal modelling, causal structures are often represented by directed acyclic graphs; however, string diagrams may equally well be used~\cite{JKZ19}.}
A circuit shape $Q$ may then be considered `at least as expressive' as a shape $P$ if it can express at least as many morphisms of the SMC.
Sometimes, this is the case simply because any circuit with shape $P$ can be rewritten into one with shape $Q$ using purely syntactical operations valid in any SMC, such as composing two or more gates, treating them as a single one; or adding trivial gates and trivial wires.

Here, we propose an order-theoretic treatment of circuit syntax and such syntactical rewrites.
We then consider the input-output \emph{connectivity} of a circuit and characterise the \emph{concept lattice}~\cite{Birk67,Wille82,GW24} as the most expressive circuit with a given connectivity.
This result has particular relevance to the search for \emph{causal decompositions} in quantum theory~\cite{LB21,pic,VMA25b}.
These are circuit decompositions of a unitary transformation where certain no-influence conditions satisfied by the unitary are reflected in the connectivity of the circuit.
While the general existence of such circuit decompositions remains an open problem, the present result demonstrates the existence of a canonical shape for those that do exist---the concept lattice---thereby answering a question posed in~\cite{LB21}.
In upcoming work~\cite{pic} this syntactic result is combined with a semantic (operator-algebraic) proof to construct causal decompositions and characterise precisely when a particular type of them exists.

This paper is structured as follows.
In \cref{sec:first-definitions}, we propose a definition of a \emph{circuit} as a partial order $P$ supplemented by two maps ${\la:\labelsA\to P}$, $\mu:\labelsB\to P$ into it.
We also introduce the relevant structure-preserving maps, which we call \emph{circuit morphisms}.
In case $\labelsA$, $\labelsB$, and $P$ are finite, the elements of $P$ can be thought of as gates in a circuit diagram, the elements of $\labelsA$ and $\labelsB$ as the overall input and output wires, and $\la,\mu$ as providing the locations of these wires in the circuit diagram.
However, all results in this work, with the exception of \cref{prop:decomposition-into-atomic-quotients}, hold in the infinite case, too.

\Cref{sec:why-rewrites} examines the structure of circuit morphisms.
We recall the theory of quotients of partially ordered sets by \emph{compatible congruences} studied in~\cite{Sturm72, Sturm73, Sturm77, KRS05,Will24} and show in \cref{thm:morphism-decomposition-2} that each circuit morphism factorises into a quotient map and a number of other elementary types of morphisms.
For the finite case, this section serves to show that circuit morphisms capture an appropriate notion of syntactical circuit rewrites, with quotient maps in particular corresponding to gate composition.

Finally, \cref{sec:connectivity} focusses on the \emph{connectivity} of a circuit, which is a binary relation $G\subseteq \labelsA\times\labelsB$ specifying which inputs can reach which outputs via directed paths through the circuit.
We recall from \textcite{Birk67} that out of any given binary relation $G\subseteq\labelsA\times\labelsB$ one may construct a complete lattice, called the \emph{concept lattice} in formal concept analysis~\cite{Wille82,GW24}; supplemented with the natural choice of maps $\la,\mu$ into it, it defines a circuit $\latt{G}$.
Our main result, \cref{thm:main-thm}, is a characterisation of $\latt{G}$ as the smallest circuit with connectivity $G$ that admits morphisms from all other circuits with the same connectivity.
In the sense described above, this shows in particular that $\latt{G}$ is most expressive, in terms of SMC morphism expressibility, among all circuits with connectivity $G$.
    \section{First definitions}\label{sec:first-definitions}
A preorder on a set $P$ is a relation $\leq\,\subseteq P\times P$ that is (i) \emph{reflexive}: $\forall p\in P: p\leq p$ and (ii) \emph{transitive}: $\forall p,q,r\in P: p\leq q\land q\leq r\implies p\leq r$; a (partial) order is one that is also (iii) \emph{antisymmetric}: $\forall p,q\in P: p\leq q \land q\leq p \implies p=q$.
We write $p<q$ iff $p\leq q$ and $p\neq q$; $p\incomp q$ iff $p\nleq q$ and $q\nleq p$; and $p\lessdot q$ iff $p$ \emph{covers} $q$ from below, that is, if $p < q$ and there is no $r\in P$ so that $p < r < q$.

\begin{definition}
    \label{def:circuit}
    Let $\labelsA$ and $\labelsB$ be sets.
    A \defn{circuit} from $\labelsA$ to $\labelsB$ is a tuple $(P,\leq,\la,\mu)$ of a set $P$, a partial order $\leq$ on $P$, and two functions $\la: \labelsA \to P$ and $\mu:\labelsB \to P$ called the \defn{input map} and \defn{output map}, respectively.
    We will often refer to elements of $P$ as \emph{gates} and to elements of $\labelsA$ and $\labelsB$ as \emph{inputs} and \emph{outputs}, respectively.
    $P$ may also refer to a circuit itself, rather than only its underlying set; in this case we write $\leq$ or $\leq_P$ and $\la_P, \mu_P$ for the order and maps constituting the circuit on $P$.
\end{definition}

\begin{definition}
    We call a circuit $P$ \defn{finite} if the set $P$ as well as the sets $\labelsA$ and $\labelsB$ are.
    Finite circuits may be conveniently depicted by diagrams.
    First of all, the \emph{Hasse diagram} of the partial order $(P,\leq)$ is the directed acyclic graph on $P$ obtained by drawing an edge from $p$ to $q$ iff $p \lessdot q$.
    The elements of $P$ are laid out on the page so that all edges point upward; the latter are then drawn as undirected edges.
    The \defn{circuit diagram} of $(P,\leq,\la,\mu)$ is obtained by supplementing the Hasse diagram with nodes for each $\lA\in\labelsA$ and $\lB\in\labelsB$ and upward-directed edges from $\lA$ to $\la(\lA)$ and from $\mu(\lB)$ to $\lB$.
    We will often leave the gates of the circuit unlabelled, in which case the diagram determines the circuit up to circuit isomorphisms (defined imminently).
\end{definition}

\begin{example}
    \label{ex:circuit-diagram}
    The diagram
    
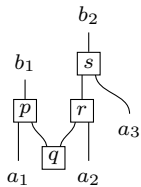
\begin{equation*}
        \begin{tikzpicture}
	\begin{pgfonlayer}{nodelayer}
		\node [style=tiny map] (0) at (0, -1) {$q$};
		\node [style=tiny map] (1) at (-0.75, 0.25) {$p$};
		\node [style=tiny map] (2) at (0.75, 0.25) {$r$};
		\node [style=none] (4) at (-0.925, 0) {};
		\node [style=none] (5) at (-0.575, 0) {};
		\node [style=none] (6) at (-0.175, -0.75) {};
		\node [style=none] (7) at (0.175, -0.75) {};
		\node [style=none] (8) at (0.575, 0) {};
		\node [style=none] (9) at (0.925, 0) {};
		\node [style=tiny map] (10) at (0.925, 1.5) {$s$};
		\node [style=none] (11) at (0.75, 1.25) {};
		\node [style=none] (12) at (1.1, 1.25) {};
		\node [style=label] (17) at (2, -0.25) {$\lA_3$};
		\node [style=label] (18) at (-0.75, 1.5) {$\lB_1$};
		\node [style=label] (19) at (0.925, 2.75) {$\lB_2$};
		\node [style=label] (21) at (0.925, -1.5) {$\lA_2$};
		\node [style=label] (22) at (-0.925, -1.5) {$\lA_1$};
	\end{pgfonlayer}
	\begin{pgfonlayer}{edgelayer}
		\draw [in=-90, out=90] (6.center) to (5.center);
		\draw [in=-90, out=90] (7.center) to (8.center);
		\draw (2) to (11.center);
		\draw [in=90, out=-90] (12.center) to (17);
		\draw (10) to (19);
		\draw (1) to (18);
		\draw (4.center) to (22);
		\draw (21) to (9.center);
	\end{pgfonlayer}
\end{tikzpicture}
    \end{equation*}
    denotes the circuit $(P,\leq,\la,\mu)$ where $P = \{p,q,r,s\}$; $\leq$ is the unique partial order satisfying $q<p$, $q<r<s$, $p\incomp r$ and $p\incomp s$;
    the sets of inputs and outputs are $\labelsA = \{\lA_1,\lA_2,\lA_3\}$ and $\labelsB = \{\lB_1,\lB_2\}$; and, finally, the input map $\la:\labelsA\to P$ acts as $\lA_1\mapsto p,\lA_2\mapsto r,\lA_3\mapsto s$ and the output map $\mu:\labelsB\to P$ as $\lB_1\mapsto p$ and $\lB_2\mapsto s$.%
    \footnote{\label{fn:transitive-reduction}Other treatments of circuits and string diagram syntax are based on directed (acyclic) graphs, or similar structures, from the outset (see e.g.\ \cite{Bon+22a}). Such an approach allows for a distinction between circuits with identical transitive reductions, i.e.\ that induce the same partial order among their gates. In \cref{ex:circuit-diagram}, for instance, one could allow for an additional wire directly from gate $q$ to gate $s$. Our order-based formalism here does not make such distinctions. They are not relevant to the question of expressivity described in the Introduction, since the wire from $q$ to $s$ can just as well be incorporated into the gate $r$.}
\end{example}

\begin{definition}
    \label{def:morphism}
    Let $P,Q$ be circuits from $\labelsA$ to $\labelsB$.
    We call a map $f: P\to Q$ a \defn{(circuit) morphism} if it satisfies the following properties.
    \begin{enumerate}[(i)]
        \item $f$ is \defn{order-preserving}: for all $p,p'\in P$, $p \leq p' \implies f(p) \leq f(p')$.
        \item $f$ \defn{respects inputs}: for all $\lA\in\labelsA$, $f(\la_P(\lA)) \geq \la_Q(\lA)$.
        \item $f$ \defn{respects outputs}: for all $\lB\in\labelsB$, $f(\mu_P(\lB)) \leq \mu_Q(\lB)$.
    \end{enumerate}
    Moreover, we say $f$ is a \defn{(circuit) embedding}, and write $f:P\hookrightarrow Q$, if it is a circuit morphism and an order embedding (i.e.\ $f(p) \leq f(p') \implies p\leq p'$).
    A \defn{(circuit) isomorphism} $f:P \xrightarrow{\sim} Q$ is a bijective circuit morphism whose inverse is also a circuit morphism.
\end{definition}

\begin{remark}
    We regard the input and output sets $\labelsA$ and $\labelsB$ as (arbitrary but) fixed; two circuits can be isomorphic only if they are defined on identical input and output sets.
    Moreover, two circuits that differ only by a permutation on $\labelsA$ or $\labelsB$ may not be isomorphic.

    Furthermore, if $f: P\to Q$ is bijective and a circuit embedding, then it is automatically an \emph{order} isomorphism (that is, an order-preserving map with order-preserving inverse); however, its inverse $f^{-1}$ is only a circuit morphism---and $f$ therefore only a \emph{circuit} isomorphism---if it also satisfies the \emph{equalities} $f\circ\la_P = \la_Q$ and $f\circ\mu_P = \mu_Q$.
    The map $f$ illustrated below is an example of a bijective circuit embedding that fails this property: it satisfies $f(\mu_P(\lB_1)) < \mu_Q(\lB_1)$, so its inverse is not a circuit morphism.
    \begin{equation}
        \label{eq:embedding-but-not-isomorphism}
        \tikzfigpad{circuits/not-isomorphism}.\qedhere
    \end{equation}
\end{remark}

We will often be interested not so much in the properties of a particular morphism $P \to Q$, but rather in whether such a morphism exists in the first place.
Since circuit morphisms can be regarded as syntatical rewriting procedures (as we will argue in the next section), the existence of such a morphism indicates that $Q$ is at least as expressive as $P$, in the sense described in the Introduction.
Therefore we define the following.

\begin{definition}\label{def:preorder}
    Denote by $\CircuitPreOrder$ the class of circuits from $A$ to $B$ preordered by the relation $\preceq$ so that $P\preceq Q$ iff there exists a morphism $P \to Q$.
    When both $P\preceq Q$ and $Q\preceq P$, we will say that $P$ and $Q$ are \defn{syntactically equivalent}.
\end{definition}

Syntactical equivalence of $P$ and $Q$ is strictly weaker than the existence of a circuit isomorphism between $P$ and $Q$; see for instance \cref{ex:syntatical-equivalence} below.
First, note the following, which is useful for showing syntactical \emph{in}equivalence.

\begin{definition}
    Let $(P,\leq,\la,\mu)$ be a circuit.
    For $p\in P$, we write
    \begin{equation}
        \label{eq:ppluspminus}
        p^- \coloneqq \{ \lA \in \labelsA : \la(\lA) \leq p \} \text{\qquad and\qquad} p^+ \coloneqq \{ \lB \in \labelsB : \mu(\lB) \geq p \}.
    \end{equation}
    That is, $p^-$ are the inputs that lie to $p$'s past, $p^+$ the outputs lying to $p$'s future.
\end{definition}

\begin{proposition}
    \label{prop:morphisms-respect-plusminus}
    If $f:P\to Q$ is a circuit morphism, then for any $p\in P$, $p^- \subseteq f(p)^-$ and $p^+ \subseteq f(p)^+$.
\end{proposition}
\begin{proof}
    Suppose $\lA\in p^-$: that is, $\la_P(\lA) \leq p$.
    Since $f$ is a circuit morphism, we have $\la_Q(\lA) \leq f(\la_P(\lA)) \leq f(p)$.
    Therefore $\lA \in f(p)^-$.
    The inclusion $p^+ \subseteq f(p)^+$ follows similarly.
\end{proof}

\begin{example}
    \label{ex:nonlocal-comp}
    Let $\labelsA = \{\lA_1,\lA_2\}$ and $\labelsB = \{\lB_1, \lB_2\}$ and consider the circuits
    \begin{equation*}
        \label{eq:nonlocal-comp}
        P \coloneqq 
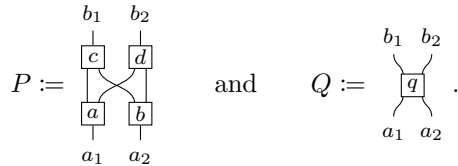
\begin{tikzpicture}
	\begin{pgfonlayer}{nodelayer}
		\node [style=tiny map] (5) at (-0.5, 0.75) {$c$};
		\node [style=tiny map] (6) at (0.75, 0.75) {$d$};
		\node [style=tiny map] (7) at (0.75, -0.75) {$b$};
		\node [style=tiny map] (8) at (-0.5, -0.75) {$a$};
		\node [style=none] (9) at (-0.65, 0.5) {};
		\node [style=none] (10) at (-0.35, 0.5) {};
		\node [style=none] (11) at (-0.65, -0.5) {};
		\node [style=none] (12) at (-0.35, -0.5) {};
		\node [style=none] (13) at (0.6, -0.5) {};
		\node [style=none] (14) at (0.9, -0.5) {};
		\node [style=none] (15) at (0.9, 0.5) {};
		\node [style=none] (16) at (0.6, 0.5) {};
		\node [style=label] (21) at (-0.5, 1.875) {$\lB_1$};
		\node [style=label] (22) at (0.75, 1.875) {$\lB_2$};
		\node [style=label] (23) at (-0.5, -1.875) {$\lA_1$};
		\node [style=label] (24) at (0.75, -1.875) {$\lA_2$};
	\end{pgfonlayer}
	\begin{pgfonlayer}{edgelayer}
		\draw (9.center) to (11.center);
		\draw [in=90, out=-90] (10.center) to (13.center);
		\draw [in=-90, out=90] (12.center) to (16.center);
		\draw (15.center) to (14.center);
		\draw (21) to (5);
		\draw (22) to (6);
		\draw (23) to (8);
		\draw (24) to (7);
	\end{pgfonlayer}
\end{tikzpicture} \qquad\text{and}\qquad Q \coloneqq \begin{tikzpicture}
	\begin{pgfonlayer}{nodelayer}
		\node [style=tiny map] (0) at (0, 0) {$q$};
		\node [style=none] (2) at (-0.25, 0.25) {};
		\node [style=none] (4) at (0.25, 0.25) {};
		\node [style=none] (6) at (-0.25, -0.25) {};
		\node [style=none] (8) at (0.25, -0.25) {};
		\node [style=label] (9) at (-0.5, 1.25) {$\lB_1$};
		\node [style=label] (10) at (0.5, 1.25) {$\lB_2$};
		\node [style=label] (11) at (-0.5, -1.25) {$\lA_1$};
		\node [style=label] (12) at (0.5, -1.25) {$\lA_2$};
	\end{pgfonlayer}
	\begin{pgfonlayer}{edgelayer}
		\draw [in=-90, out=90, looseness=1.25] (11) to (6.center);
		\draw [in=90, out=-90, looseness=1.25] (8.center) to (12);
		\draw [in=90, out=-90, looseness=1.25] (9) to (2.center);
		\draw [in=-90, out=90, looseness=1.25] (4.center) to (10);
	\end{pgfonlayer}
\end{tikzpicture}.
    \end{equation*}
    The map $f:P \to Q$ that sends all four elements $a,b,c,d$ of $P$ to $q \in Q$ is clearly a circuit morphism.
    However, there is no circuit morphism from $Q$ to $P$.
    Note that $Q$ has a gate $q$ (its only element) satisfying $q^- = \labelsA$ and $q^+ = \labelsB$.
    By \cref{prop:morphisms-respect-plusminus}, having such a gate is a property of $Q$ that must be preserved by circuit morphisms $Q\to P$; however, $P$ does not have any such gate.
    $P$ and $Q$ are therefore strictly ordered as $P \prec Q$.%
    \footnote{\label{fn:CNOT}These two circuits have particular significance to classical and quantum causation.
    In the `classical' SMC of stochastic maps, $P$ and $Q$ are equally expressive: any stochastic map between sets $\lA_1\times \lA_2$ and $\lB_1\times \lB_2$ can be expressed as the composition of four stochastic maps as in $P$.
    This relies on the ability to clone classical information and underlies the procedure of \emph{exogenisaton} of latent variables in classical causal reasoning~\cite{Evans16}.
    In quantum theory, however, where information cannot be cloned, there are bipartite quantum channels that rely on an irreducible type of \emph{interaction} between all four systems $\lA_1,\lA_2,\lB_1,\lB_2$ which cannot be simulated by circuits of the form $P$ (see e.g.~\cite{SW12}).
    This \emph{quantum} inequivalence of $P$ and $Q$, which has consequences for cryptography~\cite{Buhr+14} and holography~\cite{May19}, thus reflects the \emph{syntactical} inequivalence witnessed by the inexistence of a morphism $Q\to P$.}
\end{example}

\begin{example}
    \label{ex:syntatical-equivalence}
    The two circuits in \cref{eq:embedding-but-not-isomorphism} are not isomorphic but they do admit morphisms in both directions.
    (The morphism $f$ in \cref{eq:embedding-but-not-isomorphism} provides one from left to right; a morphism in the other direction is given by mapping both elements of the circuit on the right to the bottom element of the circuit on the left.)
    All these three circuits are therefore syntactically equivalent yet not isomorphic.
\end{example}

    \section{Circuit morphisms}\label{sec:why-rewrites}
The main purpose of this section, which is independent of our main result in \cref{sec:connectivity}, is to show that in the case of finite circuits, circuit morphisms provide a natural notion of syntactical circuit rewrites.
We do this by showing in \cref{thm:morphism-decomposition-2} that any morphism factorises into morphisms of particular elementary types, all of which can clearly be interpreted as valid syntactical rewrites in the finite case.
However, we refrain from making a formal connection with symmetric monoidal categories, keeping the focus instead on the order theory, and all notions and results, with the exception of \cref{prop:decomposition-into-atomic-quotients}, are defined and hold in the infinite case as well.
We start with some background on quotients of partial orders~\cite{Sturm72, Sturm73, Sturm77, KRS05, Will24}.

\begin{definition}
    \label{def:quotient-order}
    Let $(P,\leq)$ be a partial order and $\theta$ an equivalence relation on $P$.
    Denote by $[p]$ the equivalence class of $p\in P$ under $\theta$ and define the binary relation $\leq_\theta$ on the quotient set $P/\theta$ by
    \begin{equation}
    [p]
        \leq_\theta [q]  \quad:\!\iff\quad \exists p'\in [p],q' \in [q] \text{ such that } p' \leq q'.
    \end{equation}
    In other words, it is the minimal relation on $P/\theta$ making the quotient map $\pi_\theta : P \to P/\theta, p\mapsto [p]$ order-preserving.
\end{definition}

While the relation $\leq_\theta$ so defined is always reflexive, it may not be antisymmetric or transitive, and thus not necessarily a partial order.

\begin{example}
    \label{ex:non-compatible-congruence}
    If $P = \{p,q,r\}$, $p<q<r$, and the equivalence classes of $\theta$ are $\{p,r\}$ and $\{q\}$, then $\leq_{\theta}$ is not antisymmetric, as we have $\{p,r\} \leq_\theta \{q\} \leq_\theta \{p,r\}$.
    Moreover, if $P = \{p,q,r,s\}$, $p<q$, $r<s$, and $\{p,q\}$ and $\{r,s\}$ are $\leq$-incomparable, then the order $\leq_\theta$ defined by the equivalence relation $\theta$ with equivalence classes $\{p\},\{r,s\},\{q\}$ is not transitive.
\end{example}

The following notion has been studied in detail in Refs.~\cite{Sturm72, Sturm73, Sturm77, KRS05}, amongst other works.

\begin{definition}
    \label{def:compatible-congruence}
    An equivalence relation $\theta$ on $(P,\leq)$ is a \defn{compatible congruence}%
    \footnote{We take this terminology from~\cite{Will24,Stan86}; this has also simply been called an \emph{order congruence}~\cite{KRS05}.}
    if the transitive closure $\overrightarrow{\leq_\theta}$ of $\leq_\theta$ is antisymmetric and hence a partial order.
    When $\theta$ is a compatible congruence, we will henceforth equip the set $P/\theta$ with the order $\overrightarrow{\leq_\theta}$.
    Moreover, when ${(P,\leq,\la,\mu)}$ is a circuit, we denote by $P/\theta$ the circuit $(P/\theta,\overrightarrow{\leq_\theta},\pi_\theta\circ\la,\pi_\theta\circ\mu)$, where ${\pi_\theta:P\to P/\theta}$ is the quotient map.
\end{definition}

Which equivalences $\theta$ are compatible congruences?
First, call a subset $S$ of $P$ \defn{convex} if for all $s,t\in S$ and $p\in P$, $s\leq p\leq t \implies p\in S$.
If $\theta$ is a compatible congruence on $P$, then every equivalence class $[p]$ is necessarily convex~\cite[\S36]{Sturm72}; this is required to avoid failure of antisymmetry of $\leq_\theta$ (cf.\ \cref{ex:non-compatible-congruence}).
However, not every equivalence relation $\theta$ on $P$ whose equivalence classes are convex is a compatible congruence.
As an example, consider the order underlying the circuit $P$ from \cref{ex:nonlocal-comp} and the equivalence relation $\theta$ whose equivalence classes are $\{a,d\}$ and $\{b,c\}$.
Once one identifies $a$ and $d$, the equivalence class $\{b,c\}$ becomes non-convex, resulting in a relation $\leq_\theta$ that is not antisymmetric.
However, as long as one performs a sequence of quotients that identify the elements of a \emph{single} convex subset at a time, the resulting relation will remain a partial order.
Every quotient by a compatible congruence in fact arises in this way.

\begin{definition}
    Let $S \subseteq P$ be a convex set.
    Denote by $\theta_S$ the equivalence relation on $P$ that has equivalence classes $S$ and $\{p\}$ for $p\in P\setminus S$.
    If $S$ contains precisely two elements, that is, $S = \{p,q\}$ where $p\lessdot q$ or $p \incomp q$, then we call $\theta_S$ an \defn{atomic compatible congruence}.
\end{definition}

\begin{proposition}
    \label{prop:decomposition-into-atomic-quotients}
    If $S\subseteq P$ is convex, then $\theta_S$ is a compatible congruence.
    Moreover, if $P$ is finite then an equivalence relation $\theta$ on $P$ is a compatible congruence if and only if the quotient map $\pi_\theta: P \to P/\theta$ is, up to order isomorphisms, a finite composition of quotients by atomic compatible congruences.
\end{proposition}

We omit the (straightforward) proof (see also~\cite[\S7]{Sturm73}), the main point being that each quotient by a compatible congruence can be regarded as the result of a `gate composition' procedure in a circuit diagram, with each atomic compatible congruence $\theta_{\{p,q\}}$ corresponding to a single sequential composition (if $p\lessdot q$ or $q\lessdot p$) or a single parallel composition (monoidal product; if $p \incomp q$).
The connection with circuit morphisms is due to the following.

\begin{definition}
    The \defn{kernel} of a map $f: P\to Q$, denoted $\ker f$, is the equivalence relation on $P$ whose equivalence classes are the fibres of $f$---that is, $p (\ker f) p' \iff f(p) = f(p')$ for $p,p'\in P$.
\end{definition}
\begin{proposition}[{\cite[\S2]{Sturm77},~\cite[Prop.~4.10]{Will24}}]
    An equivalence $\theta$ on $P$ is a compatible congruence if and only if it is the kernel of an order-preserving map $f: P \to Q$ for some partially ordered set $Q$.%
\end{proposition}

\begin{proposition}
    \label{prop:morphism-decomposition-1}
    Let $f : P \to Q$ be a circuit morphism.
    Let $\pi_f : P \twoheadrightarrow P/\ker f$ be the quotient map and $[f] : P/\ker f \to Q$ the (well-defined) injective map $[p]\mapsto f(p)$.
    Then $\pi_f$ and $[f]$ are circuit morphisms and $f = [f] \circ \pi_f$:
    \begin{equation}
        \begin{tikzpicture}
	\begin{pgfonlayer}{nodelayer}
		\node [style=object] (0) at (0, 1.25) {$P$};
		\node [style=object] (1) at (3.5, 1.25) {$P/\ker f$};
		\node [style=object] (2) at (3.5, -1.5) {$Q$};
	\end{pgfonlayer}
	\begin{pgfonlayer}{edgelayer}
		\draw [style=dir, ->>] (0) to node [auto, font={\footnotesize}] {$\pi_f$} (1);
		\draw [style=dir] (1) to node [auto, font={\footnotesize}] {$[f]$} (2);
		\draw [style=dir] (0) to node [below left, font={\footnotesize}] {$f$} (2);
	\end{pgfonlayer}
\end{tikzpicture}
.
    \end{equation}
\end{proposition}
\iflongversion
\begin{proof}
    The only nontrivial statement to verify is that $[f] : P/\ker f \to Q$ is a circuit morphism.
    Recall that $P/\ker f$ is ordered by $\overrightarrow{\leq_{\ker f}}$ and has input and output maps $\la_{P/\ker f} = \pi_f \circ \la_P$ and $\mu_{P/\ker f} = \pi_f \circ \mu_P$, respectively.
    To show that $[f]$ is order-preserving, first suppose that $[p],[q]\in P/\ker f$ and $[p] \leq_{\ker f} [q]$.
    This means that there exist $p',q'\in P$ so that $[p]=[p']$, $p'\leq_P q'$, and $[q']=[q]$.
    By construction, then,
    \begin{equation*}
    [f]([p])
        = f(p) = f(p') \leq_Q f(q') = f(q) = [f]([q]).
    \end{equation*}
    Thus, $[p] \leq_{\ker f} [q] \implies [f]([p]) \leq_Q [f]([q])$.
    Since $\leq_Q$ is transitive, we also have $[p] \mathrel{\overrightarrow{\leq_{\ker f}}} [q] \implies [f]([p]) \leq_Q [f]([q])$, so $[f]$ is order-preserving.
    That it also respects inputs is immediate: $[f] \circ \la_{P/\ker f} = [f] \circ \pi_f \circ \la_P = f \circ \la_P \geq \la_Q$ pointwise for all $\lA\in\labelsA$, where the last inequality follows from the fact that $f$ is itself a circuit morphism.
    Similarly, $[f]$ respects outputs.
\end{proof}
\else

\fi

It remains to argue that in addition to the quotient map $\pi_f$, the injective circuit morphism $[f]$ also models a valid syntactical circuit rewrite.%
\footnote{Note that in disanalogy to the isomorphism theorems for algebraic structures, the map $[f] : P/\ker f \to \im(f) \subseteq Q$ is not necessarily a circuit isomorphism; it is injective but not always an order embedding.}

\begin{definition}
    \label{def:injective-morphism-types}
    The following are three types of injective circuit morphisms.
    Each corresponds to a certain elementary type of circuit rewrite and is exemplified below.
    \begin{enumerate}[(i)]
        \item\label{itm:adds-isolated-gates} Let $(P,\leq_P,\la,\mu)$ be a circuit from $\labelsA$ to $\labelsB$ and $S$ an arbitrary set. Write $P\sqcup S$ for the disjoint union of the sets $P$ and $S$.
        Let $\leq$ be the partial order on $P\sqcup S$ that satisfies $\leq \cap\, (P\times P) =\, \leq_P$ and $s\incomp q$ for any $s\in S$ and $q\in P\sqcup S$.
        Denote by $P\sqcup S$ the circuit from $\labelsA$ to $\labelsB$ given by $(P\sqcup S, \leq, \la,\mu)$.
        Say that a circuit morphism $f:P \to Q$ \defn{adds isolated gates} if $Q = P\sqcup S$ for some set $S$ and $f$ is the inclusion map $P \hookrightarrow P\sqcup S$.
        Each such map is in fact a circuit embedding.
        \item\label{itm:adds-wires} A circuit morphism that \defn{adds wires} is one of the form $\id_P: (P,{\leq},\la,\mu) \to (P, {\leq'}, \la,\mu)$, $p\mapsto p$, where $\leq,\leq'$ are such that $p\leq q \implies p \leq' q$.
        \item\label{itm:reguides} A circuit morphism that \defn{advances inputs and delays outputs} is one of the form $\id_P: (P,{\leq},\la,\mu) \to (P,{\leq},\la',\mu')$, $p\mapsto p$ where $\la(\lA) \geq \la'(\lA)$ for all $\lA\in\labelsA$ and $\mu(\lB) \leq \mu'(\lB)$ for all $\lB\in\labelsB$.
        This is always an embedding.\qedhere
    \end{enumerate}
\end{definition}

\begin{example}
    The first morphism below is a quotient map; the other three are of the types just introduced.
    
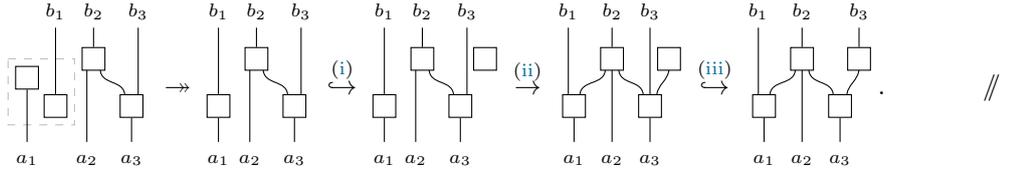
\begin{equation*}
        \begin{tikzpicture}
	\begin{pgfonlayer}{nodelayer}
		\node [style=tiny map] (0) at (-1.5, 0.25) {};
		\node [style=tiny map] (1) at (-0.75, -0.5) {};
		\node [style=tiny map] (4) at (0.25, 0.75) {};
		\node [style=tiny map] (5) at (1.25, -0.5) {};
		\node [style=none] (6) at (0.075, 0.5) {};
		\node [style=none] (7) at (0.425, 0.5) {};
		\node [style=none] (8) at (1.075, -0.25) {};
		\node [style=none] (9) at (1.425, -0.25) {};
		\node [style=small label] (14) at (-1.5, -1.875) {$\lA_1$};
		\node [style=small label] (15) at (0.075, -1.875) {$\lA_2$};
		\node [style=small label] (16) at (1.25, -1.875) {$\lA_3$};
		\node [style=small label] (17) at (-0.75, 2) {$\lB_1$};
		\node [style=small label] (18) at (0.25, 2) {$\lB_2$};
		\node [style=small label] (19) at (1.425, 2) {$\lB_3$};
		\node [style=none] (20) at (-2, 0.75) {};
		\node [style=none] (21) at (-0.25, 0.75) {};
		\node [style=none] (22) at (-0.25, -1) {};
		\node [style=none] (23) at (-2, -1) {};
	\end{pgfonlayer}
	\begin{pgfonlayer}{edgelayer}
		\draw [in=90, out=-90, looseness=1.25] (7.center) to (8.center);
		\draw [style=merge] (20.center)
			 to (21.center)
			 to (22.center)
			 to (23.center)
			 to cycle;
		\draw (14) to (0);
		\draw (15) to (6.center);
		\draw (16) to (5);
		\draw (1) to (17);
		\draw (4) to (18);
		\draw (9.center) to (19);
	\end{pgfonlayer}
\end{tikzpicture} \twoheadrightarrow \begin{tikzpicture}
	\begin{pgfonlayer}{nodelayer}
		\node [style=tiny map] (24) at (-0.825, -0.5) {};
		\node [style=tiny map] (27) at (0.175, 0.75) {};
		\node [style=tiny map] (28) at (1.175, -0.5) {};
		\node [style=none] (29) at (0, 0.5) {};
		\node [style=none] (30) at (0.35, 0.5) {};
		\node [style=none] (31) at (1, -0.25) {};
		\node [style=none] (32) at (1.35, -0.25) {};
		\node [style=small label] (37) at (-0.825, -1.875) {$\lA_1$};
		\node [style=small label] (38) at (0, -1.875) {$\lA_2$};
		\node [style=small label] (39) at (1.175, -1.875) {$\lA_3$};
		\node [style=small label] (40) at (-0.825, 2) {$\lB_1$};
		\node [style=small label] (41) at (0.175, 2) {$\lB_2$};
		\node [style=small label] (42) at (1.35, 2) {$\lB_3$};
	\end{pgfonlayer}
	\begin{pgfonlayer}{edgelayer}
		\draw [in=90, out=-90, looseness=1.25] (30.center) to (31.center);
		\draw (24) to (40);
		\draw (27) to (41);
		\draw (32.center) to (42);
		\draw (38) to (29.center);
		\draw (24) to (37);
		\draw (28) to (39);
	\end{pgfonlayer}
\end{tikzpicture} \stackrel{\text{(\ref{itm:adds-isolated-gates})}}{\hookrightarrow} \begin{tikzpicture}
	\begin{pgfonlayer}{nodelayer}
		\node [style=tiny map] (0) at (-0.9, -0.5) {};
		\node [style=tiny map] (3) at (0.1, 0.75) {};
		\node [style=tiny map] (4) at (1.1, -0.5) {};
		\node [style=none] (5) at (-0.075, 0.5) {};
		\node [style=none] (6) at (0.275, 0.5) {};
		\node [style=none] (7) at (0.925, -0.25) {};
		\node [style=none] (8) at (1.275, -0.25) {};
		\node [style=small label] (13) at (-0.9, -1.875) {$\lA_1$};
		\node [style=small label] (14) at (-0.075, -1.875) {$\lA_2$};
		\node [style=small label] (15) at (1.1, -1.875) {$\lA_3$};
		\node [style=small label] (16) at (-0.9, 2) {$\lB_1$};
		\node [style=small label] (17) at (0.1, 2) {$\lB_2$};
		\node [style=small label] (18) at (1.275, 2) {$\lB_3$};
		\node [style=tiny map] (19) at (1.75, 0.75) {};
	\end{pgfonlayer}
	\begin{pgfonlayer}{edgelayer}
		\draw [in=90, out=-90, looseness=1.25] (6.center) to (7.center);
		\draw (0) to (16);
		\draw (3) to (17);
		\draw (8.center) to (18);
		\draw (14) to (5.center);
		\draw (15) to (4);
		\draw (13) to (0);
	\end{pgfonlayer}
\end{tikzpicture} \ \stackrel{\text{(\ref{itm:adds-wires})}}{\rightarrow} \begin{tikzpicture}
	\begin{pgfonlayer}{nodelayer}
		\node [style=tiny map] (0) at (-0.9, -0.5) {};
		\node [style=tiny map] (3) at (0.1, 0.75) {};
		\node [style=tiny map] (4) at (1.1, -0.5) {};
		\node [style=none] (5) at (-0.075, 0.5) {};
		\node [style=none] (6) at (0.275, 0.5) {};
		\node [style=none] (7) at (0.925, -0.25) {};
		\node [style=none] (8) at (1.1, -0.25) {};
		\node [style=small label] (13) at (-0.9, -1.875) {$\lA_1$};
		\node [style=small label] (14) at (0.1, -1.875) {$\lA_2$};
		\node [style=small label] (15) at (1.1, -1.875) {$\lA_3$};
		\node [style=small label] (16) at (-1.05, 2) {$\lB_1$};
		\node [style=small label] (17) at (0.1, 2) {$\lB_2$};
		\node [style=small label] (18) at (1.1, 2) {$\lB_3$};
		\node [style=none] (19) at (1.275, -0.25) {};
		\node [style=tiny map] (20) at (1.6, 0.75) {};
		\node [style=none] (21) at (1.6, 0.5) {};
		\node [style=none] (22) at (-0.75, -0.25) {};
		\node [style=none] (23) at (-1.05, -0.25) {};
		\node [style=none] (24) at (0.1, 0.5) {};
	\end{pgfonlayer}
	\begin{pgfonlayer}{edgelayer}
		\draw [in=90, out=-90, looseness=1.25] (6.center) to (7.center);
		\draw [in=-90, out=90] (19.center) to (21.center);
		\draw [in=-90, out=90, looseness=1.25] (22.center) to (5.center);
		\draw (3) to (17);
		\draw (8.center) to (18);
		\draw (23.center) to (16);
		\draw (0) to (13);
		\draw (14) to (24.center);
		\draw (4) to (15);
	\end{pgfonlayer}
\end{tikzpicture} \ \stackrel{\text{(\ref{itm:reguides})}}{\hookrightarrow} \begin{tikzpicture}
	\begin{pgfonlayer}{nodelayer}
		\node [style=tiny map] (0) at (-0.75, -0.5) {};
		\node [style=tiny map] (3) at (0.25, 0.75) {};
		\node [style=tiny map] (4) at (1.25, -0.5) {};
		\node [style=none] (5) at (0.075, 0.5) {};
		\node [style=none] (6) at (0.425, 0.5) {};
		\node [style=none] (7) at (1.075, -0.25) {};
		\node [style=none] (8) at (1.25, -0.25) {};
		\node [style=small label] (13) at (-0.75, -1.875) {$\lA_1$};
		\node [style=small label] (14) at (0.25, -1.875) {$\lA_2$};
		\node [style=small label] (15) at (1.25, -1.875) {$\lA_3$};
		\node [style=small label] (16) at (-0.9, 2) {$\lB_1$};
		\node [style=small label] (17) at (0.25, 2) {$\lB_2$};
		\node [style=small label] (18) at (1.75, 2) {$\lB_3$};
		\node [style=none] (19) at (1.425, -0.25) {};
		\node [style=tiny map] (20) at (1.75, 0.75) {};
		\node [style=none] (21) at (1.75, 0.5) {};
		\node [style=none] (22) at (-0.6, -0.25) {};
		\node [style=none] (23) at (-0.9, -0.25) {};
		\node [style=none] (24) at (0.25, 0.5) {};
	\end{pgfonlayer}
	\begin{pgfonlayer}{edgelayer}
		\draw [in=90, out=-90, looseness=1.25] (6.center) to (7.center);
		\draw [in=-90, out=90] (19.center) to (21.center);
		\draw [in=-90, out=90, looseness=1.25] (22.center) to (5.center);
		\draw (0) to (13);
		\draw (14) to (24.center);
		\draw (4) to (15);
		\draw (3) to (17);
		\draw (23.center) to (16);
		\draw (20) to (18);
	\end{pgfonlayer}
\end{tikzpicture}.
        \qedhere
    \end{equation*}
\end{example}

\begin{theorem}
    \label{thm:morphism-decomposition-2}
    A map $f: P \to Q$ is a circuit morphism if and only if it is, up to isomorphisms, the composition of a finite sequence of maps of the following four types.
    \begin{itemize}
        \item quotients by compatible congruences (which, in case $P$ is finite, are in turn always finite compositions of quotients by atomic compatible congruences, by \cref{prop:decomposition-into-atomic-quotients});
        \item morphisms that add isolated gates;
        \item morphisms that add wires;
        \item morphisms that advance inputs and delay outputs.
    \end{itemize}
\end{theorem}
\iflongversion
\begin{proof}
    Let $f$ be a circuit morphism.
    By \cref{prop:morphism-decomposition-1} it suffices to show that $[f]: P/\ker f \to Q$ factorises into morphisms of the latter three types.
    Denote by $\tilde Q$ the circuit $(Q,\leq_Q, f\circ\la_P, f\circ\mu_P)$ and consider
    \begin{align*}
        \begin{split}
            \pi_f: P &\twoheadrightarrow P/\ker f \\
            p &\mapsto [p];
        \end{split} \\
        \begin{split}
            \id_{P/\ker f}: P/\ker f &\hookrightarrow P/\ker f \sqcup (Q\setminus\im(f)) \\
            \quad [p] &\mapsto [p];
        \end{split} \\
        \begin{split}
            g : P/\ker f \sqcup (Q\setminus\im(f)) &\to \tilde Q \\
            [p] &\mapsto [f]([p]) \quad\text{ for } p\in P; \\
            q &\mapsto q \quad\text{ for } q\in Q\setminus\im(f);
        \end{split} \\
        \begin{split}
            \id_Q : \tilde Q &\hookrightarrow Q \\
            q &\mapsto q.
        \end{split}
    \end{align*}
    Recall from \cref{def:injective-morphism-types} that we consider $Q\setminus\im(f)$ as discretely ordered.
    It is clear that all four maps are, up to isomorphism, morphisms of the required types and that $f$ factorises as
    \begin{equation*}
        \begin{tikzpicture}
	\begin{pgfonlayer}{nodelayer}
		\node [style=object] (0) at (-2.25, 0) {$P$};
		\node [style=object] (1) at (1, 0) {$P/\ker f$};
		\node [style=object] (2) at (9.25, 0) {\mbox{$P/\ker f \sqcup$ \\  $(Q\setminus\im(f))$ }};
		\node [style=object] (3) at (15, 0) {$\tilde Q$};
		\node [style=object] (4) at (18.75, 0) {$Q$};
	\end{pgfonlayer}
	\begin{pgfonlayer}{edgelayer}
		\draw [style=dir, ->>] (0) to node [auto, font={\footnotesize}] {$\pi_f$} (1);
		\draw [style=dir, right hook->] (1) to node [auto, font={\footnotesize}] {$\id_{P/\ker f}$} (2);
		\draw [style=dir] (2) to node [auto, font={\footnotesize}] {$g$} (3);
		\draw [style=dir, right hook->] (3) to node [auto, font={\footnotesize}] {$\id_Q$} (4);
		\draw [style=dir, bend right, looseness=0.75] (0) to node [auto, font={\footnotesize}] {$f$} (4);
		\draw [style=dir, bend left, looseness=0.75] (1) to node [auto, font={\footnotesize}] {$[f]$} (4);
	\end{pgfonlayer}
\end{tikzpicture}
.
        \qedhere
    \end{equation*}
\end{proof}
\fi

\begin{remark}
    In addition to syntactical circuit rewrites, the morphisms from \cref{def:morphism} can also be used to model implementations of circuits in a spacetime.
    Let $(M,g)$ be a time-oriented Lorentzian manifold without closed causal curves, so that the order $\preceq_g$ on $M$ induced by the metric $g$ and the time orientation is a partial order.
    Let $\labelsA$ and $\labelsB$ be finite sets of physical systems.
    Functions $\la_M : \labelsA \to M$ and $\mu_M : \labelsB \to M$ provide localisations of these systems at points in spacetime, and together with the spacetime constitute a tuple $(M,\preceq_g, \la_M,\mu_M)$ which in our terminology is just a special `circuit', in this case with infinite underlying set $M$.
    A morphism $f:P\to M$ from a finite circuit $P$ into $M$ provides a localisation of each of $P$'s gates at points in spacetime in such a way that appropriate causal curves are guaranteed to exist for each $P$'s wires (edges of its circuit diagram) and inputs and outputs are picked up and dropped off at the locations requested by $\la_M$ and $\mu_M$.
    See also Refs.~\cite{SalzS24,losec}.
\end{remark}
    \section{Connectivity}\label{sec:connectivity}
The following notion arises in quantum causality, in particular in the study of causal decompositions of unitary transformations~\cite{LB21,pic,VMA25b}.

\begin{definition}\label{def:connectivity}
    Let $P$ be a circuit from $\labelsA$ to $\labelsB$.
    The \defn{connectivity} of $P$ is the binary relation $G_P \subseteq \labelsA\times\labelsB$ satisfying
    \begin{equation}
        \forall \lA\in\labelsA, \lB\in\labelsB : \quad \lA \relG_P \lB \iff \la_P(\lA) \leq \mu_P(\lB).
    \end{equation}
    This defines a map $G_{(-)} : \CircuitPreOrder \to \powset(\labelsA\times\labelsB), P \mapsto G_P$.
    When the powerset $\powset(\labelsA\times\labelsB)$ is ordered under inclusion, this map is order-preserving.
\end{definition}

This last claim is simply to verify.
It tells us that, as should be expected, syntactical circuit rewrites never decrease connectivity.
As a consequence, if two circuits have distinct connectivity, this immediately indicates that they are syntactically inequivalent (see \cref{def:preorder}).
Also circuits with identical connectivity may however be inequivalent:

\begin{example}
    \label{ex:ccc}
    If $\cccG\subseteq\labelsA\times\labelsB$ is the relation between $\labelsA \coloneqq \{\lA_1,\lA_2,\lA_3\}$ and $\labelsB \coloneqq \{\lB_1,\lB_2,\lB_3\}$ given by
    \begin{equation*}
        \label{eq:ccc-relation}
        \cccG\ \coloneqq \begin{tikzpicture}
	\begin{pgfonlayer}{nodelayer}
		\node [style=label] (0) at (0, -1) {$\lA_1$};
		\node [style=label] (1) at (2, -1) {$\lA_2$};
		\node [style=label] (2) at (4, -1) {$\lA_3$};
		\node [style=label] (3) at (0, 1) {$\lB_1$};
		\node [style=label] (4) at (2, 1) {$\lB_2$};
		\node [style=label] (5) at (4, 1) {$\lB_3$};
	\end{pgfonlayer}
	\begin{pgfonlayer}{edgelayer}
		\draw [style=dir] (0) to (3);
		\draw [style=dir] (0) to (4);
		\draw [style=dir] (1) to (3);
		\draw [style=dir] (1) to (4);
		\draw [style=dir] (1) to (5);
		\draw [style=dir] (2) to (5);
		\draw [style=dir] (2) to (4);
	\end{pgfonlayer}
\end{tikzpicture}
,
    \end{equation*}
    then the circuits
    \begin{equation*}
        \label{eq:ccc-circuits}
        \cls{\cccG}\ \coloneqq \begin{tikzpicture}
	\begin{pgfonlayer}{nodelayer}
		\node [style=tiny map] (0) at (-1.75, 1.125) {};
		\node [style=tiny map] (1) at (0, 1.125) {};
		\node [style=tiny map] (2) at (1.75, 1.125) {};
		\node [style=tiny map] (3) at (-1.75, -1.125) {};
		\node [style=tiny map] (4) at (0, -1.125) {};
		\node [style=tiny map] (5) at (1.75, -1.125) {};
		\node [style=none] (6) at (-2, 0.875) {};
		\node [style=none] (7) at (-1.5, 0.875) {};
		\node [style=none] (8) at (-0.25, 0.875) {};
		\node [style=none] (9) at (0, 0.875) {};
		\node [style=none] (10) at (0.25, 0.875) {};
		\node [style=none] (11) at (1.5, 0.875) {};
		\node [style=none] (12) at (2, 0.875) {};
		\node [style=none] (13) at (-2, -0.875) {};
		\node [style=none] (14) at (-1.5, -0.875) {};
		\node [style=none] (15) at (-0.25, -0.875) {};
		\node [style=none] (16) at (0, -0.875) {};
		\node [style=none] (17) at (0.25, -0.875) {};
		\node [style=none] (18) at (1.5, -0.875) {};
		\node [style=none] (19) at (2, -0.875) {};
		\node [style=label] (26) at (-1.75, -2.625) {$\lA_1$};
		\node [style=label] (27) at (0, -2.625) {$\lA_2$};
		\node [style=label] (28) at (1.75, -2.625) {$\lA_3$};
		\node [style=label] (29) at (-1.75, 2.625) {$\lB_1$};
		\node [style=label] (30) at (0, 2.625) {$\lB_2$};
		\node [style=label] (31) at (1.75, 2.625) {$\lB_3$};
	\end{pgfonlayer}
	\begin{pgfonlayer}{edgelayer}
		\draw (6.center) to (13.center);
		\draw [in=-90, out=90] (14.center) to (8.center);
		\draw [in=90, out=-90] (7.center) to (15.center);
		\draw (16.center) to (9.center);
		\draw [in=90, out=-90] (10.center) to (18.center);
		\draw [in=-90, out=90] (17.center) to (11.center);
		\draw (12.center) to (19.center);
		\draw (26) to (3);
		\draw (27) to (4);
		\draw (28) to (5);
		\draw (2) to (31);
		\draw (1) to (30);
		\draw (0) to (29);
	\end{pgfonlayer}
\end{tikzpicture} \text{\quad and\quad} \latt{\cccG}\ \coloneqq \begin{tikzpicture}
	\begin{pgfonlayer}{nodelayer}
		\node [style=tiny map] (0) at (1.75, 1.25) {};
		\node [style=none] (1) at (1.5, 1) {};
		\node [style=none] (2) at (2, 1) {};
		\node [style=none] (3) at (0.5, 0.25) {};
		\node [style=none] (4) at (1, 0.25) {};
		\node [style=none] (5) at (2.5, 0.25) {};
		\node [style=none] (6) at (3, 0.25) {};
		\node [style=label] (8) at (1.75, 2.5) {$\lB_2$};
		\node [style=label] (11) at (3, 2.5) {$\lB_3$};
		\node [style=label] (12) at (0.5, 2.5) {$\lB_1$};
		\node [style=tiny map] (13) at (1.75, -1.25) {};
		\node [style=tiny map] (14) at (0.75, 0) {};
		\node [style=tiny map] (15) at (2.75, 0) {};
		\node [style=none] (16) at (1.5, -1) {};
		\node [style=none] (17) at (2, -1) {};
		\node [style=none] (18) at (0.5, -0.25) {};
		\node [style=none] (19) at (1, -0.25) {};
		\node [style=none] (20) at (2.5, -0.25) {};
		\node [style=none] (21) at (3, -0.25) {};
		\node [style=label] (23) at (1.75, -2.5) {$\lA_2$};
		\node [style=label] (26) at (3, -2.5) {$\lA_3$};
		\node [style=label] (27) at (0.5, -2.5) {$\lA_1$};
	\end{pgfonlayer}
	\begin{pgfonlayer}{edgelayer}
		\draw [in=-90, out=90] (4.center) to (1.center);
		\draw [in=90, out=-90] (2.center) to (5.center);
		\draw [in=90, out=-90] (19.center) to (16.center);
		\draw [in=-90, out=90] (17.center) to (20.center);
		\draw (13) to (23);
		\draw (26) to (21.center);
		\draw (27) to (18.center);
		\draw (0) to (8);
		\draw (11) to (6.center);
		\draw (12) to (3.center);
	\end{pgfonlayer}
\end{tikzpicture}
    \end{equation*}
    both have connectivity $\cccG$.
    There exists a circuit morphism $f: \cls{\cccG} \to \latt{\cccG}$: it is the quotient by the equivalence relation whose equivalence classes are highlighted below.%
    \footnote{If you thought `Why does \latt{\cccG} have no wire directly from its bottom gate to its top gate?', please refer back to \cref{fn:transitive-reduction}.}
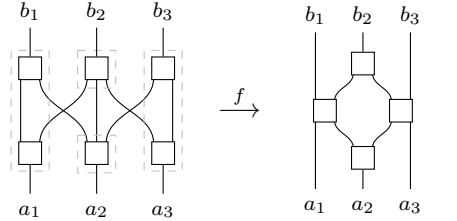
    \begin{equation*}
        \tikzfigpad{circuits/ccc-classical-merge} \stackrel{f}{\longrightarrow} \tikzfigpad{circuits/ccc-diamond}.
    \end{equation*}
    There is no circuit morphism in the other direction, however; therefore $\cls{\cccG} \prec \latt{\cccG}$.
    This follows by an application of \cref{prop:morphisms-respect-plusminus} similar to that in \cref{ex:nonlocal-comp}.
    Indeed, there exist quantum channels that may be expressed by circuits of the form $\latt{\cccG}$ but not by ones of the form $\cls{\cccG}$ (cf.\ \cref{fn:CNOT}).
\end{example}

In \cref{subsec:concept-circuit}, we construct the analogue of $\latt{\cccG}$ for general relations $G\subseteq\labelsA\times\labelsB$ and show that it in fact admits incoming morphisms from \emph{all} other circuits with connectivity $G$.
We generalise $\cls{\cccG}$ in \cref{subsec:primitive-circuit}, constructing a circuit $\cls{G}$ which satisfies the dual property, admitting outgoing morphisms into any other circuit with connectivity $G$.

\subsection{The concept lattice}\label{subsec:concept-circuit}
Let $\labelsA$ and $\labelsB$ be arbitrary sets.
From each binary relation $G\subseteq\labelsA\times\labelsB$ one may construct a complete lattice via a construction first realised by \textcite{Birk67} and extensively studied in formal concept analysis~\cite{Wille82,GW24}.
We use the terminology from~\cite{Wille82,GW24} and call this the \emph{concept lattice}.
We will briefly recall this construction and then show, in \cref{thm:main-thm}, that with the natural choice of input and output maps, it constitutes the smallest circuit with connectivity $G$ that admits incoming morphisms from every other circuit with the same connectivity.

We will use the notation
\begin{equation}
    \label{eq:parental-and-filial-sets}
    \ChG(\lA) \coloneqq \{\lB\in\labelsB \mid \lA\relG \lB\}  \text{\quad for } \lA\in\labelsA \text{\quad and\quad}
    \PaG(\lB) \coloneqq \{\lA\in\labelsA \mid \lA\relG \lB\}  \text{\quad for } \lB\in\labelsB.
\end{equation}
Denote by $\powset(\labelsA)$ and $\powset(\labelsB)$ the powersets of $\labelsA$ and $\labelsB$ and consider the maps%
\footnote{Our choice of letters $\gC$ and $\gP$ reflects the origins of this work in quantum causal modelling, where $G$ encodes \emph{causal structure} and the elements of $\ChG(\lA) = \gC(\{\lA\})$ are referred to as the `(causal) children' of $\lA\in\labelsA$, while those of $\PaG(\lB) = \gP(\{\lB\})$ are the `(causal) parents' of $\lB\in\labelsB$.}
\begin{align}
    \gC &: \powset(\labelsA) \to \powset(\labelsB) \dblcolon \a \mapsto \bigcap_{\lA\in\a} \ChG(\lA) = \{\lB \in \labelsB \mid \forall \lA \in\a: \lA \relG \lB \} \\
    \text{and\quad} \gP &: \powset(\labelsB) \to \powset(\labelsA) \dblcolon \b \mapsto \bigcap_{\lB\in\b} \PaG(\lB) = \{\lA \in \labelsA \mid \forall \lB \in\b: \lA \relG \lB \}. \label{eq:galois-connection-p}
\end{align}
In particular, $\gC(\emptyset) = \labelsB$ and $\gP(\emptyset) = \labelsA$.
With $\powset(\labelsA)$ and $\powset(\labelsB)$ ordered under inclusion, these maps are order-reversing:
\begin{equation}
    \label{eq:p-c-order-reversing}
    \begin{split}
        \forall \a,\a'\subseteq\labelsA&: \a\subseteq\a'\implies \gC(\a)\supseteq \gC(\a') \\
        \text{ and \quad }\forall \b,\b'\subseteq\labelsB&: \b\subseteq\b'\implies \gP(\b)\supseteq \gP(\b').
    \end{split}
\end{equation}
They also satisfy
\begin{equation}
    \label{eq:extensivity}
    \forall \a\subseteq\labelsA: \a \subseteq \gP\gC(\a) \qquad\text{and}\qquad \forall\b\subseteq\labelsB: \b\subseteq \gC\gP(\b).
\end{equation}
\cref{eq:p-c-order-reversing,eq:extensivity} make the pair $(\gC,\gP)$ into an (antitone) \emph{Galois connection} between the partially ordered sets $\powset(\labelsA)$ and $\powset(\labelsB)$~\cite{Ore44}.%
\footnote{In categorical terms, these maps form a dual adjunction between the orders $\powset(\labelsA)$ and $\powset(\labelsB)$ seen as categories.}
A direct consequence is that
\begin{equation}
    \label{eq:galois-connection}
    \gP\gC\gP = \gP \qquad\text{and}\qquad \gC\gP\gC = \gC.
\end{equation}
Moreover, the map $\gP\gC: \powset(\labelsA)\to\powset(\labelsA)$ acts as a \emph{closure operator} on $\powset(\labelsA)$, meaning it is
\begin{enumerate}[(i)]
    \item \label{itm:closure-operator-extensive} extensive: $\forall \a\subseteq\labelsA: \a \subseteq \gP\gC(\a)$;
    \item \label{itm:closure-operator-idempotent} idempotent: $\gP\gC\gP\gC = \gP\gC$; and
    \item \label{itm:closure-operator-order-preserving} order-preserving: $\forall\a,\a'\subseteq\labelsA:\a\subseteq\a' \implies \gP\gC(\a) \subseteq \gP\gC(\a')$.
\end{enumerate}
We therefore say that $\gP\gC(\a)$ is the \defn{closure} of $\a\subseteq\labelsA$ and that $\a$ is \defn{closed} if $\gP\gC(\a) = \a$.
It follows from (\ref{itm:closure-operator-extensive})--(\ref{itm:closure-operator-order-preserving}) that the closure of a set is precisely the smallest closed set containing it and that the intersection of (arbitrarily many) closed sets is again closed.
Moreover, by virtue of \cref{eq:galois-connection}, the closed subsets of $\labelsA$ are precisely the sets $\a$ that lie in the image of $\gP$.

\begin{example}
    With respect to the relation $\cccG$ defined in \cref{ex:ccc}, the closure of $\{\lA_1\}\subseteq \labelsA$ is $\{\lA_1,\lA_2\}$, while $\{\lA_2\}$ is itself closed.
\end{example}

Let $\lattA \subseteq\powset(\labelsA)$ be the set of closed subsets of $\labelsA$.
As $\lattA$ is closed under arbitrary intersections, it forms a complete lattice when ordered under inclusion: that is, any subset of elements $\{\a_i\}_{i\in I}\subseteq \lattA$ has a greatest lower bound and a least upper bound in $\lattA$.
They are given, respectively, by
\begin{equation*}
    \bigmeet_{i\in I} \a_i = \bigcap_{i\in I} \a_i  \qquad\text{and}\qquad  \bigjoin_{i\in I} \a_i = \gP\gC\left( \bigcup_{i\in I} \a_i \right).
\end{equation*}

Dually, the map $\gC\gP: \powset(\labelsB)\to\powset(\labelsB)$ satisfies conditions analogous to (\ref{itm:closure-operator-extensive})--(\ref{itm:closure-operator-order-preserving}) above and is therefore a closure operator on $\powset(\labelsB)$.
Let $\lattB$ be the set of closed subsets of $\labelsB$ with respect to this closure operator.
Just like $\lattA$, $(\lattB,\subseteq)$ forms a complete lattice; so does the opposite order $\lattBop \coloneqq (\lattB, \supseteq)$.
Note that $\lattA$ and $\lattB$ are the images of the operators $\gP$ and $\gC$, respectively; therefore, \cref{eq:galois-connection} implies that $\gC$ and $\gP$ restrict to a pair of inverse maps $\gC : \lattA \to \lattBop$ and $\gP : \lattBop \to \lattA$.
These are order isomorphisms, by \cref{eq:p-c-order-reversing}, and thus isomorphisms of complete lattices.
The concept lattice, defined as follows, takes a symmetric approach, incorporating both $\lattA$ and $\lattBop$.

\begin{definition}[\cite{Wille82,GW24}]
    Let $G\subseteq \labelsA\times\labelsB$.
    The \defn{concept lattice}%
    \footnote{In formal concept analysis, $\labelsA$ is a set of \emph{objects}, $\labelsB$ a set of \emph{properties}, and $G$ the relation specifiying which objects have which properties; together, this triple is referred to as a \emph{formal context}.
    An element $\<\a,\b\>\in\latt{G}$ is then called a \emph{concept}: it defines a set of objects ($\a$) that is uniquely characterised by the set of properties ($\b$) that they share (i.e.\ $\a=\gP(\b)$ and $\b=\gC(\a)$).}
    $\latt{G}$ is the set
    \begin{equation}
        \latt{G} \coloneqq \{\<\a,\b\> \in \lattA\times \lattB \mid \a = \gP(\b) \text{ and } \b = \gC(\a)\}
    \end{equation}
    ordered by
    \begin{equation}
        \<\a,\b\> \leq \<\a',\b'\> \quad:\!\iff\quad \a \subseteq \a' \text{ or, equivalently, } \b \supseteq \b'.
    \end{equation}
    For arbitrary $v\in\latt{G}$, we will denote by $\a_v\subseteq\labelsA$ and $\b_v\subseteq\labelsB$ the closed sets so that $v = \<\a_v,\b_v\>$.
    $\latt{G}$ forms a complete lattice whose greatest lower bounds and least upper bounds are given by
    \begin{align}
        \bigmeet_{i\in I} v_i = \left\< \bigcap_{i\in I} \a_{v_i}, \gC\gP\left(\bigcup_{i\in I} \b_{v_i}\right) \right\> \qquad\text{and}\qquad
        \bigjoin_{i\in I} v_i = \left\< \gP\gC\left(\bigcup_{i\in I} \a_{v_i}\right), \bigcap_{i\in I} \b_{v_i} \right\>.
    \end{align}
    Finally, define the input and output maps
    \label{eq:la-and-mu}
    \begin{align}
        \la_\latt{G}: \labelsA \to \latt{G}, \quad \lA &\mapsto \< \gP\gC(\{\lA\}), \gC(\{\lA\})\> \\
        \text{ and \quad }\mu_\latt{G}: \labelsB \to \latt{G}, \quad \lB &\mapsto \< \gP(\{\lB\}), \gC\gP(\{\lB\})\>.
    \end{align}
    That is, $\la_\latt{G}(\lA)$ is the smallest element $\<\a,\b\> \in \latt{G}$ such that $\lA\in\a$ and, dually, $\mu_\latt{G}(\lB)$ is the greatest element $\<\a,\b\>\in\latt{G}$ so that $\lB\in\b$.
    The tuple $(\latt{G},\leq,\la_\latt{G}, \mu_\latt{G})$ so defined constitutes a circuit, which we will also refer to as the concept lattice (with respect to the relation $G$).
\end{definition}

\begin{remark}
    $\lattA$, $\lattBop$, and $\latt{G}$ defined above provide three different, isomorphic perspectives on the {\canonicalcircuit}.
    Denote by $\pi_\labelsA: \latt{G} \to \lattA$ the map $v \mapsto \a_v$; this is an order isomorphism whose inverse is $\a \mapsto \<\a,\gC(\a)\>$.
    If we equip $\lattA$ with the input and output maps $\la_{\lattA} \coloneqq \pi_\labelsA \circ \la_\latt{G} = \gP\gC(\{-\})$ and $\mu_{\lattA} \coloneqq \pi_\labelsA \circ \mu_\latt{G} = \gC(\{-\})$, then $\pi_\labelsA$ becomes a circuit isomorphism.
    We can do the same to $\lattBop$ using the map $\pi_\labelsB : \latt{G}\to \lattBop$ defined by $v \mapsto \b_v$, which has inverse $\b\mapsto\<\gP(\b),\b\>$.
    The following diagram of circuit isomorphisms then commutes.
    \begin{equation*}
        \begin{tikzpicture}
	\begin{pgfonlayer}{nodelayer}
		\node [style=object] (0) at (0, 1.75) {$\latt{G}$};
		\node [style=object] (1) at (-3, -1.25) {$\lattA$};
		\node [style=object] (2) at (3, -1.25) {$\lattBop$};
	\end{pgfonlayer}
	\begin{pgfonlayer}{edgelayer}
		\draw [style=dir] (0) to node [above left, font={\footnotesize}] {$\pi_\labelsA$} (1);
		\draw [style=dir] (0) to node [above right, font={\footnotesize}] {$\pi_\labelsB$} (2);
		\draw [style=dir, bend left=15] (1) to node [above, font={\footnotesize}] {$\gC$} (2);
		\draw [style=dir, bend left=15] (2) to node [below, font={\footnotesize}] {$\gP$} (1);
	\end{pgfonlayer}
\end{tikzpicture}
.\qedhere
    \end{equation*}
\end{remark}

\begin{example}
    The circuit $\latt{\cccG}$ in \cref{ex:ccc} is the {\canonicalcircuit} for the relation $\cccG$.
    As another example, consider the relation on $\labelsA = \{1,2,3,4\}$ and $\labelsB = \{x,y,z\}$ given by
    \begin{equation*}
        G\ = \begin{tikzpicture}
	\begin{pgfonlayer}{nodelayer}
		\node [style=label] (0) at (0, -1) {1};
		\node [style=label] (1) at (2, -1) {2};
		\node [style=label] (2) at (4, -1) {3};
		\node [style=label] (3) at (1.25, 1) {$x$};
		\node [style=label] (4) at (3, 1) {$y$};
		\node [style=label] (5) at (4.75, 1) {$z$};
		\node [style=label] (6) at (6, -1) {4};
	\end{pgfonlayer}
	\begin{pgfonlayer}{edgelayer}
		\draw [style=dir] (0) to (3);
		\draw [style=dir] (1) to (3);
		\draw [style=dir] (1) to (4);
		\draw [style=dir] (1) to (5);
		\draw [style=dir] (2) to (5);
		\draw [style=dir] (2) to (4);
		\draw [style=dir] (6) to (5);
	\end{pgfonlayer}
\end{tikzpicture}

    \end{equation*}
    The closed subsets of $\labelsA$ are $\labelsA$, all sets of the form $\PaG(\lB)$ for some $\lB\in\labelsB$, and their intersections: $\lattA = \{1234, 12, 23, 234, 2\}$.
    The lattice $\lattBop$ is constructed similarly; their Hasse diagrams are given by
    \begin{equation*}
        \lattA \quad=\tikzfigpad{lattice-example-LA} \quad\text{and}\qquad \lattBop \quad=\quad\begin{tikzpicture}
	\begin{pgfonlayer}{nodelayer}
		\node [style=label] (0) at (-0.25, -2.25) {$xyz$};
		\node [style=label] (1) at (-1.75, -0.75) {$x$};
		\node [style=label] (2) at (1.25, -0.75) {$yz$};
		\node [style=label] (3) at (2.75, 0.75) {$z$};
		\node [style=label] (4) at (1.25, 2.25) {$\emptyset$};
	\end{pgfonlayer}
	\begin{pgfonlayer}{edgelayer}
		\draw (1) to (4);
		\draw (4) to (3);
		\draw (2) to (3);
		\draw (0) to (2);
		\draw (0) to (1);
	\end{pgfonlayer}
\end{tikzpicture}
.
    \end{equation*}
    The resulting {\canonicalcircuit} is%
    \footnote{When interpreting $\latt{G}$ as a circuit in a symmetric monoidal category which is \emph{causal} or satisfies \emph{process terminality}, in the terminology of~\cite{PQP,KHC17}, then the top gate $\<1234,\emptyset\>$ in this example, which has no outputs, would necessarily correspond to a discarding operation. It may therefore be omitted without any reduction in expressivity of the circuit shape. However, we make no process terminality assumption in our current general setting.}
    \begin{equation*}
        \latt{G} \quad = \tikzfigpad{lattice-example-L}. \iflongversion\qedhere\fi
    \end{equation*}
\end{example}

\iflongversion
Note that in the above example, the inputs in the past of each gate $\<\a,\b\>$ are precisely the elements of $\a$, and the outputs in its future the elements of $\b$.
That is the case in general:

\begin{proposition}
    \label{prop:la-mu-properties}
    Let $v\in\latt{G}$.
    For any $\lA\in\labelsA,\lB\in\labelsB$, we have
    \begin{equation}
        \la_\latt{G}(\lA) \leq v \iff \lA\in\a_v \quad\text{and}\quad v \leq \mu_{\latt{G}}(\lB) \iff \lB\in\b_v.
    \end{equation}
    \iflongversion
    In other words, using the notation of \cref{eq:ppluspminus}, $v^- = \a_v$ and $v^+ = \b_v$.
    \fi
    Moreover,
    \begin{equation}
        \label{eq:join-and-meet-dense}
        v = \bigjoin_{\lA\in\a_v} \la_\latt{G}(\lA) = \bigmeet_{\lB\in\b_v} \mu_{\latt{G}}(\lB).
    \end{equation}
\end{proposition}
\begin{proof}
    We have
    \begin{equation*}
        \la_\latt{G}(\lA) \leq v \iff \gP\gC(\{\lA\}) \subseteq \a_v \iff \{\lA\}\subseteq \a_v \iff \lA\in\a_v.
    \end{equation*}
    The second equivalence here follows from extensitivity~\eqref{eq:extensivity} (left to right) and idempotency (right to left) of the closure operator $\gP\gC$, as well as the fact that $\a_v$ is closed.

    As a consequence, $v$ is an upper bound to the set $\{\la_\latt{G}(\lA) \mid \lA\in\a_v\}$ and must therefore be at least as large as the least upper bound $\bigjoin_{\lA\in\a_v} \la_\latt{G}(\lA)$.
    The converse statement that $v \leq \bigjoin_{\lA\in\a_v} \la_\latt{G}(\lA)$ is equivalent to saying that $\a_v \subseteq \gP\gC \left(\bigcup_{\lA\in\a_v} \gP\gC(\{\lA\})\right)$, which is immediate from extensitivity of $\gP\gC$.
    The dual statements follow similarly.
\end{proof}
\fi

\begin{proposition}
    \label{prop:trivial-endomorphism-group}
    Each circuit endomorphism of $\latt{G}$ is trivial: $\End(\latt{G}) = \{\id_{\latt{G}}\}$.
\end{proposition}
\begin{proof}
    Let $f: \latt{G}\to\latt{G}$ be a circuit morphism.
    Since $f$ is order-preserving, we have, for all $S\subseteq \latt{G}$,
    \begin{equation*}
        f\left(\bigjoin S\right) \geq \bigjoin f(S)  \text{\qquad and \qquad} f\left(\bigmeet S\right) \leq \bigmeet f(S).
    \end{equation*}
    Moreover, since $f$ is a circuit morphism, it satisfies $f(\la_\latt{G}(\lA)) \geq \la_\latt{G} (\lA)$ and $f(\mu_\latt{G}(\lB)) \leq \mu_\latt{G}(\lB)$ for all $\lA\in\labelsA$, $\lB\in\labelsB$.
    Combining these facts \iflongversion with \cref{prop:la-mu-properties} \fi gives that for $v\in \latt{G}$,
    \begin{equation*}
        f(v) = f\left( \bigjoin_{\lA\in\a_v} \la_\latt{G}(\lA) \right) \geq \bigjoin_{\lA\in\a_v} f\left(\la_{\latt{G}}(\lA)\right) \geq \bigjoin_{\lA\in\a_v} \la_{\latt{G}}(\lA) = v
    \end{equation*}
    and,
    \iflongversion dually,
    \begin{equation*}
        f(v) = f\left( \bigmeet_{\lB\in\b_v} \mu_\latt{G}(\lB) \right) \leq \bigmeet_{\lB\in\b_v} f\left(\mu_{\latt{G}}(\lB)\right) \leq \bigmeet_{\lB\in\b_v} \mu_{\latt{G}}(\lB) = v.
    \end{equation*}
    \fi
    Therefore $f(v) = v$.
\end{proof}

\iflongversion
\begin{proposition}
    \label{prop:any-lattice-admits-incoming-morphisms}
    Let $L$ be a circuit from $\labelsA$ to $\labelsB$ which has connectivity $G$ and whose underlying partial order $(L,\leq_L)$ is a complete lattice.
    Then $L$ admits incoming morphisms $f:P \to L$ from any other circuit $P$ with connectivity $G_P \subseteq G$.
\end{proposition}
\begin{proof}
    Let $P$ be a circuit with connectivity $G_P \subseteq G$ and consider the map
    \begin{equation}
        \label{eq:map-into-general-lattice}
        f: P \to L, \quad p \mapsto \bigjoin_{\lA: \la_P(\lA) \leq p} \la_L(\lA) = \bigjoin_{\lA \in p^-} \la_L(\lA),
    \end{equation}
    where $\bigjoin$ is the join in $L$ and $p^-$ is as in \cref{eq:ppluspminus}.
    This map is order-preserving:
    \begin{equation*}
        p \leq_P q \implies p^- \subseteq q^- \implies \bigjoin_{\lA \in p^-} \la_L(\lA) \leq \bigjoin_{\lA \in q^-} \la_L(\lA).
    \end{equation*}
    It also respects inputs: for all $\lA\in\labelsA$,
    \begin{equation*}
        f(\la_P(\lA)) = \bigjoin_{\lA': \la_P(\lA') \leq \la_P(\lA)} \la_L(\lA') \geq \la_L(\lA).
    \end{equation*}
    Finally, $f$ respects outputs: since $G_P\subseteq G_L$, we have $\la_P(\lA) \leq \mu_P(\lB) \implies \la_L(\lA) \leq \mu_L(\lB)$, so that for all $\lB\in\labelsB$,
    \begin{equation*}
        f(\mu_P(\lB)) = \bigjoin_{\lA:\, \la_P(\lA) \leq \mu_P(\lB)} \la_L(\lA) \leq \bigjoin_{\lA:\, \la_L(\lA) \leq \mu_L(\lB)} \la_L(\lA) \leq \mu_L(\lB).
    \end{equation*}
    We conclude that $f:P\to L$ is a circuit morphism.
    A dual construction, which instead sets $f(p) \coloneqq \bigmeet_{\lB: \mu_P(\lB)\geq p} \mu_L(\lB)$, yields a generally distinct circuit morphism.
\end{proof}
\fi

\iflongversion
The previous three propositions allow us to prove our main result.
\fi

\begin{theorem}
    \label{thm:main-thm}
    Let $\labelsA$ and $\labelsB$ be sets and $G\subseteq \labelsA\times\labelsB$ a relation.
    The concept lattice $\latt{G}$ is, up to circuit isomorphisms, the unique circuit satisfying the following three properties.
    \begin{enumerate}[(i)]
        \item\label{itm:main-thm-connectivity} $\latt{G}$ has connectivity $G_{\latt{G}} = G$.
        \item\label{itm:main-thm-rewrite} For any circuit $P$ with connectivity $G_P\subseteq G$, there exists a morphism $f: P \to \latt{G}$.\footnote{Changing the inclusion $G_P\subseteq G$ in \itmref{itm:main-thm-rewrite} to an equality would not change the validity of the theorem and while it might seem to make the uniqueness statement stronger, it is not difficult to see that this is in fact not the case.}
        \item\label{itm:main-thm-smallest} If $K$ is any other circuit satisfying (\ref{itm:main-thm-connectivity}) and (\ref{itm:main-thm-rewrite}) then there is an injective morphism $\latt{G} \to K$.
    \end{enumerate}
    In fact, for any circuit $Q$ with connectivity $G$ and morphism $g:\latt{G} \to Q$, $g$ is an embedding.
\end{theorem}
\begin{proof}
    The fact that $\latt{G}$ has connectivity $G$ is a restatement of a basic result in formal concept analysis~\cite[Thm.~3]{GW24} that follows
    directly from \cref{prop:la-mu-properties}:
    \begin{equation*}
        \la_\latt{G}(\lA) \leq \mu_\latt{G}(\lB) \iff \lA\in\a_{\mu_\latt{G}(\lB)} \iff \lA\in \gP(\{\lB\}) \iff \lA\relG\lB.
    \end{equation*}
    That \itmref{itm:main-thm-rewrite} is satisfied by $\latt{G}$ is a direct consequence of \cref{prop:any-lattice-admits-incoming-morphisms} above and the fact that $\latt{G}$ is a complete lattice.

    We now prove the final claim in the theorem statement, which implies (\ref{itm:main-thm-smallest}).
    Let $Q$ be a circuit with connectivity $G$ and $g:\latt{G} \to Q$ be a morphism.
    By \itmref{itm:main-thm-rewrite}, there also exists a morphism $f:Q\to\latt{G}$.
    The composite $f\circ g$ is an endomorphism of $\latt{G}$ and is therefore, by \cref{prop:trivial-endomorphism-group}, the identity morphism $\id_{\latt{G}}$.
    Thus, $g$ has a left-inverse $f$ which is, in particular, order-preserving; this means that $g$ must be an order-embedding and thus a circuit embedding.

    It remains to show that $\latt{G}$ is the unique circuit, up to isomorphism, satisfying (\ref{itm:main-thm-connectivity})--(\ref{itm:main-thm-smallest}).
    Suppose that $K$ also satisfies these properties; then there exist morphisms $f: K \to \latt{G}$ and $g: \latt{G} \to K$ such that $f$ and $g$ are both injective.
    The composite $f \circ g : \latt{G} \to \latt{G}$ is an endomorphism of $\latt{G}$ and therefore, again, equal to the identity $\id_{\latt{G}}$.
    Therefore $f$ must be surjective; since it is also injective, it is a bijective circuit morphism whose inverse $g$ is also a circuit morphism.
    In other words, $f$ is a circuit isomorphism.
\end{proof}

In terms of the preorder from \cref{def:preorder}---i.e.\ when forgetting the distinction between parallel morphisms---\cref{thm:main-thm} leads to the following.
\begin{corollary}\label{cor:concept-lattice}
    For any circuit $P$ and relation $G\subseteq \labelsA\times\labelsB$, we have
    \begin{equation}
        G_P \subseteq G \iff P \preceq \latt{G}.
    \end{equation}
    That is, the map $\latt{(-)} : \powset(\labelsA\times\labelsB) \to \CircuitPreOrder$ is order-preserving and forms an upper Galois adjoint to the map $G_{(-)} : \CircuitPreOrder \to \powset(\labelsA\times\labelsB)$.
\end{corollary}
This determines the concept lattice up to syntactical equivalence but not up to circuit isomorphism; information about the circuit morphisms (regarding e.g.\ their injectivity) as in \cref{thm:main-thm} is required for the latter.%
\footnote{It is tempting to try to formulate \cref{thm:main-thm} itself as the characterisation of $\latt{G}$ as some right adjoint or appropriate colimit in the (non-preorder) category of circuits and circuit morphisms, just as \cref{cor:concept-lattice} does for the preorder $\CircuitPreOrder$; however, this seems not possible to do in a meaningful way, as the non-preorder category has too many parallel morphisms.}

\begin{remark}
    The Basic Theorem of Formal Concept Analysis~\cite{Wille82},~\cite[Thm.~3]{GW24} tells us that any complete lattice $(L,\leq)$ arises as the concept lattice of \emph{some} binary relation $G$.
    More precisely, a circuit ${(L,\leq,\la,\mu)}$ whose underlying order $(L,\leq)$ is a complete lattice is circuit-isomorphic to the {\canonicalcircuit} $L_G$ iff the connectivity of $L$ is $G$ and the images $\la(\labelsA), \mu(\labelsB) \subseteq L$ are join- and meet-dense in $L$, respectively.
\end{remark}


\iflongversion

\subsection{The {\primitivecircuit}}\label{subsec:primitive-circuit}
To close, we will construct another circuit with a given connectivity relation $G\subseteq\labelsA\times\labelsB$, denoted by $\cls{G}$, which generalises the circuit $\cls{\cccG}$ from \cref{ex:ccc}.
Dually to $\latt{G}$, it admits morphisms into (rather than from) all other circuits with connectivity $G$.
There is one obvious candidate for a generalisation of $\cls{\cccG}$: it is the circuit $(\labelsA\sqcup\labelsB, \leq, \iota_\labelsA,\iota_\labelsB)$ whose underlying set is the disjoint union $\labelsA\sqcup\labelsB$; whose order $\leq$ satisfies $p < q$ iff $p\in\labelsA, q\in\labelsB$ and $pGq$; and whose input and output maps are the inclusion maps $\iota_\labelsA:\labelsA\to\labelsA\sqcup\labelsB$ and $\iota_\labelsB:\labelsB\to\labelsA\sqcup\labelsB$, respectively.
For a general relation $G$, however, this circuit may contain unnecessarily many gates.
We would like $\cls{G}$ to be the \emph{smallest} circuit with connectivity $G$ admitting morphisms into all other circuits with the same connectivity, i.e.\ satisfying a condition analogous to \cref{thm:main-thm}\itmref{itm:main-thm-smallest}.
This requires taking some edge cases into account, leading us to the following definition.

\begin{definition}
    \label{def:cls-circuit}
    Let $G\subseteq\labelsA\times\labelsB$ and define
    \begin{align}
        \cAg &\coloneqq \{ \lA\in\labelsA \mid \card{\ChG(\lA)} \neq 1 \} \text{\quad and }\\
        \cBg &\coloneqq \{ \lB\in\labelsB \mid \card{\PaG(\lB)} \neq 1 \text{ or } \forall \lA\in\PaG(\lB) : \card{\ChG(\lA)} = 1 \}.
    \end{align}
    The \defn{\primitivecircuit} with connectivity $G$ is the circuit $\cls{G} = (\cAg\sqcup\cBg, \leq, \lap, \mup)$, where $\leq$ is the partial order satisfying $p<q$ iff $p \in \cAg, q\in \cBg$, and $p\relG q$; and where $\lap:\labelsA\to\cAg\sqcup\cBg$ and $\mup:\labelsB\to\cAg\sqcup\cBg$ are the maps defined by
    \begin{equation}
        \lap(\lA) = \begin{cases}
                        \lA & \text{ if } \lA\in\cAg \\ \lB & \text{ if } \ChG(\lA) = \{\lB\}
        \end{cases}
        \text{\quad and\quad}
        \mup(\lB) = \begin{cases}
                        \lB & \text{ if } \lB\in\cBg \\ \lA & \text{ if } \PaG(\lB) = \{\lA\}
        \end{cases}.
        \qedhere
    \end{equation}
\end{definition}

\begin{example}
    If $\labelsA = \{\lA_1,\lA_2,\lA_3,\lA_4,\lA_5\}$, $\labelsB = \{\lB_1,\lB_2,\lB_3,\lB_4,\lB_5\}$, and
    \begin{equation*}
        G\ = 
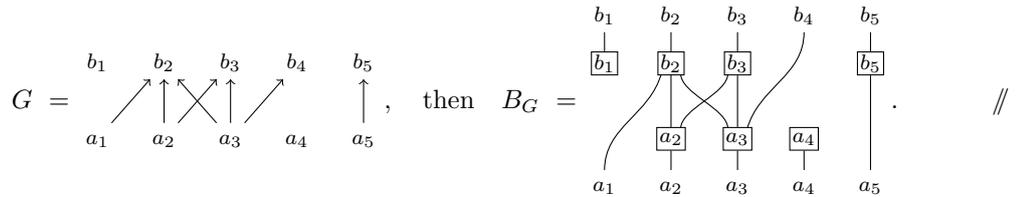
\begin{tikzpicture}
	\begin{pgfonlayer}{nodelayer}
		\node [style=label] (0) at (0, -1) {$\lA_1$};
		\node [style=label] (1) at (1.75, -1) {$\lA_2$};
		\node [style=label] (2) at (3.5, -1) {$\lA_3$};
		\node [style=label] (3) at (5.25, -1) {$\lA_4$};
		\node [style=label] (4) at (7, -1) {$\lA_5$};
		\node [style=label] (5) at (0, 1) {$\lB_1$};
		\node [style=label] (6) at (1.75, 1) {$\lB_2$};
		\node [style=label] (7) at (3.5, 1) {$\lB_3$};
		\node [style=label] (8) at (5.25, 1) {$\lB_4$};
		\node [style=label] (9) at (7, 1) {$\lB_5$};
	\end{pgfonlayer}
	\begin{pgfonlayer}{edgelayer}
		\draw [style=dir] (1) to (6);
		\draw [style=dir] (1) to (7);
		\draw [style=dir] (2) to (7);
		\draw [style=dir] (2) to (6);
		\draw [style=dir] (0) to (6);
		\draw [style=dir] (2) to (8);
		\draw [style=dir] (4) to (9);
	\end{pgfonlayer}
\end{tikzpicture}
, \quad\text{then}\quad \cls{G}\ = \begin{tikzpicture}
	\begin{pgfonlayer}{nodelayer}
		\node [style=tiny map] (0) at (1.75, -1) {$\lA_2$};
		\node [style=tiny map] (1) at (3.5, -1) {$\lA_3$};
		\node [style=tiny map] (2) at (5.25, -1) {$\lA_4$};
		\node [style=tiny map] (3) at (0, 1) {$\lB_1$};
		\node [style=tiny map] (4) at (1.75, 1) {$\lB_2$};
		\node [style=tiny map] (5) at (3.5, 1) {$\lB_3$};
		\node [style=tiny map] (6) at (7, 1) {$\lB_5$};
		\node [style=label] (7) at (0, -2.25) {$\lA_1$};
		\node [style=label] (8) at (1.75, -2.25) {$\lA_2$};
		\node [style=label] (9) at (3.5, -2.25) {$\lA_3$};
		\node [style=label] (10) at (5.25, -2.25) {$\lA_4$};
		\node [style=label] (11) at (7, -2.25) {$\lA_5$};
		\node [style=label] (12) at (0, 2.25) {$\lB_1$};
		\node [style=label] (13) at (1.75, 2.25) {$\lB_2$};
		\node [style=label] (14) at (3.5, 2.25) {$\lB_3$};
		\node [style=label] (15) at (5.25, 2.25) {$\lB_4$};
		\node [style=label] (16) at (7, 2.25) {$\lB_5$};
		\node [style=none] (17) at (1.5, 0.75) {};
		\node [style=none] (18) at (1.75, 0.75) {};
		\node [style=none] (19) at (2, 0.75) {};
		\node [style=none] (20) at (1.75, -0.75) {};
		\node [style=none] (21) at (2, -0.75) {};
		\node [style=none] (22) at (3.25, 0.75) {};
		\node [style=none] (23) at (3.25, -0.75) {};
		\node [style=none] (24) at (3.75, -0.75) {};
	\end{pgfonlayer}
	\begin{pgfonlayer}{edgelayer}
		\draw (1) to (5);
		\draw (6) to (16);
		\draw (2) to (10);
		\draw (1) to (9);
		\draw (0) to (8);
		\draw (3) to (12);
		\draw (4) to (13);
		\draw (5) to (14);
		\draw [in=-90, out=90, looseness=0.75] (21.center) to (22.center);
		\draw [in=-90, out=90] (20.center) to (18.center);
		\draw [in=90, out=-90, looseness=0.75] (19.center) to (23.center);
		\draw [in=90, out=-105] (17.center) to (7);
		\draw [in=-90, out=75] (24.center) to (15);
		\draw (6) to (11);
	\end{pgfonlayer}
\end{tikzpicture}
.
        \qedhere
    \end{equation*}
\end{example}

We have the following analogues of \cref{prop:trivial-endomorphism-group} and \cref{thm:main-thm}.
(\cref{prop:cls-trivial-endomorphism-group} is generally not satisfied by the circuit on $\labelsA\sqcup\labelsB$ discussed above \cref{def:cls-circuit}.)

\begin{proposition}
    \label{prop:cls-trivial-endomorphism-group}
    Each circuit endomorphism of $\cls{G}$ is trivial: $\End(\cls{G}) = \{\id_{\cls{G}}\}$.
\end{proposition}
\begin{proof}
    Let $f:\cls{G}\to\cls{G}$ be a circuit morphism and let $\lA\in\cAg$.
    We show that $f(\lA) = \lA$.
    First of all, we have $f(\lA) = f(\lap(\lA)) \geq \lap(\lA) = \lA$.
    Suppose for contradiction that this inequality is strict: that $f(\lA) > \lA$.
    By construction of $\cls{G}$, this means that $f(\lA) = \lB$ for some $\lB\in\ChG(\lA)\cap\cBg$.
    Now, take any $\lB'\in\ChG(\lA)$.
    Then $\lap(\lA) \leq \mup(\lB')$; this is a consequence of the fact that $\cls{G}$ has connectivity $G$, which can be straightforwardly verified.
    We therefore have $\lB = f(\lA) = f(\lap(\lA)) \leq f(\mup(\lB')) \leq \mup(\lB')$.
    Since $\lB$ is by construction a maximal element of $\cls{G}$, we must in fact have $\lB = \mup(\lB')$, which implies $\lB = \lB'$.
    We conclude that $\card{\ChG(\lA)} = 1$, contradicting the assumption that $\lA\in\cAg$.
    Therefore $f(\lA) = \lA$.

    Next, suppose that $\lB\in\cBg$. We wish to show that $f(\lB) = \lB$.
    To begin with, we have $f(\lB) = f(\mup(\lB)) \leq \mup(\lB) = \lB$.
    Suppose that in fact $f(\lB) < \lB$.
    Then $f(\lB) = \lA$ for some $\lA\in\PaG(\lB)\cap\cAg$.
    Take any $\lA' \in \PaG(\lB)$; we have $\lap(\lA') \leq \mup(\lB)$, so $\lap(\lA') \leq f(\lap(\lA')) \leq f(\mup(\lB)) = f(\lB) = \lA$, which by construction means that $\lA' = \lA$.
    Therefore, $\card{\PaG(\lB)} = 1$.
    Since $\lB\in\cBg$ by assumption, we must have $\card{\ChG(\lA)} = 1$.
    This is in contradiction with the fact that $\lA\in\cAg$.
    Therefore $f(\lB) = \lB$.
\end{proof}

\begin{theorem}
    \label{thm:primitive-theorem}
    Let $\labelsA$ and $\labelsB$ be sets and $G\subseteq\labelsA\times\labelsB$ a relation.
    Then $\cls{G}$ is, up to circuit isomorphisms, the unique circuit from $\labelsA$ to $\labelsB$ satisfying the following three properties.
    \begin{enumerate}[(i)]
        \item\label{itm:cls-thm-connectivity} $\cls{G}$ has connectivity $G_{\cls{G}} = G$.
        \item\label{itm:cls-thm-rewrite} For any circuit $P$ with connectivity $G_P \supseteq G$, there exists a morphism $f: \cls{G} \to P$.
        \item\label{itm:cls-thm-smallest} If $Q$ is any other circuit satisfying (\ref{itm:cls-thm-connectivity}) and (\ref{itm:cls-thm-rewrite}) then there is an injective morphism $\cls{G} \to Q$.
    \end{enumerate}
    In fact, for any circuit $Q$ satisfying \itmref{itm:cls-thm-connectivity} and \itmref{itm:cls-thm-rewrite} and morphism $f:\cls{G} \to Q$, $f$ is an embedding.
\end{theorem}
\begin{proof}
    We only show that \itmref{itm:cls-thm-rewrite} holds for $\cls{G}$.
    \itmref{itm:cls-thm-connectivity} can be straightforwardly verified and all other claims of the theorem follow in a way similar to the proof of \cref{thm:main-thm}, now using \cref{prop:cls-trivial-endomorphism-group} instead of \cref{prop:trivial-endomorphism-group}.

    To show \itmref{itm:cls-thm-rewrite}, let $P$ be a circuit with connectivity $G_P \supseteq G$.
    Define the map $f: \cls{G} \to P$ by
    \begin{align*}
        f(\lA) &\coloneqq \la_P(\lA)  \text{\quad for } \lA\in\cAg; \\
        f(\lB) &\coloneqq \mu_P(\lB)  \text{\quad for } \lB\in\cBg.
    \end{align*}
    It follows from the construction of $\cls{G}$ and the connectivity of $P$ that $f$ is order-preserving.
    $f$ also respects inputs: if $\lA\in\cAg$ then $f(\lap(\lA)) = f(\lA) = \la_P(\lA)$, while if $\lA\in\labelsA\setminus\cAg$, then $\ChG(\lA) = \{\lB\}$ for some $\lB\in\cBg$, so that $f(\lap(\lA)) = f(\lB) = \mu_P(\lB) \geq \la_P(\lA)$.
    The final inequality here is due to the connectivity of $P$.
    Similarly, $f(\mup(\lB))\leq \mu_P(\lB)$ for all $\lB\in\labelsB$.
    In conclusion, $f$ is a circuit morphism $\cls{G}\to P$.
\end{proof}

Finally, analogously to \cref{cor:concept-lattice}, we can summarise the implications of the above theorem for the preorder on $\CircuitPreOrder$ as follows.
\begin{corollary}\label{cor:basic-circuit}
    For any circuit $P$ and relation $G\subseteq \labelsA\times\labelsB$, we have
    \begin{equation}
        G\subseteq G_P \iff \cls{G} \preceq P.
    \end{equation}
    In other words, the map $\cls{(-)} : \powset(\labelsA\times\labelsB) \to \CircuitPreOrder$ is order-preserving and forms a lower Galois adjoint to $G_{(-)} : \CircuitPreOrder \to \powset(\labelsA\times\labelsB)$.
\end{corollary}

Combining \cref{cor:concept-lattice,cor:basic-circuit}, we see that for any circuit $P$ and relation $G$,
\begin{equation}
    G_P = G  \iff  \cls{G} \preceq P \preceq \latt{G}.
\end{equation}

\fi
    \section{Closing remarks}\label{sec:conclusion}
\paragraph{Relation to other works}
While the treatment of input-output connectivity and the resulting introduction of lattice theory into the study of circuit syntax is, as far as I am aware, new to this work, the combinatorics of circuit syntax more generally has of course been investigated previously.
Most other approaches are based on (directed) graph-like structures, rather than partial orders.
\textcite{Bon+22a}, for example, present a comprehensive formalism for string diagram rewrites based on discrete cospans of directed hypergraphs, establishing a formal equivalence with string diagrams in free symmetric monoidal categories.
Their requirement that rewrites be \emph{convex} seems to parallel the emergence of \emph{compatible} congruences as kernels of circuit morphisms in this work.

Another order-theoretic framework similar to the one here was proposed in recent independent work of~\textcite{SalzS24}.
Their definitions of circuits and morphisms (there called \emph{framed partial orders} and \emph{frame- and order-preserving maps}) are distinct from, but seem for most practical purposes equivalent to our definitions in \cref{sec:first-definitions}.
(Precisely, for any given pair of sets $\labelsA,\labelsB$ the category proposed in~\textcite{SalzS24} is isomorphic to a strict subcategory of our circuits and circuit morphisms; but the approaches become equivalent after quotienting by syntactical equivalence.)
Like us, they are motivated by considering the (in)existence of morphisms as a means of comparing expressivity of circuit shapes.
The primary focus of~\cite{SalzS24} is on the notion of syntactical equivalence and the interaction of circuits with spacetime structure.
In particular, the authors show, restricting to the finite case, that every equivalence class $[P]$ contains a circuit $S\in[P]$ such that all morphisms $S\to P'\in[P]$ are embeddings; they call $S$, which is unique up to isomorphisms, the \emph{minimal representative} of the class $[P]$.
This result also holds in the formalism proposed here,
\iflongversion
and from~\cref{thm:main-thm,thm:primitive-theorem} we see that $\latt{G}$ and $\cls{G}$ are particular examples of such minimal representatives.
\fi

\paragraph{Future directions}
The results presented here have relevance to quantum causality.
\Cref{thm:main-thm} shows that the concept lattice acts as a canonical shape for causal decompositions of unitary transformations, studied in~\cite{LB21,VMA25b,pic}.
By doing so it provides the syntax of the circuit decomposition that one should be looking for; what remains is to find appropriate semantics in terms of quantum systems and quantum gates.
This problem boils down to a close interplay between order theory and finite-dimensional operator algebra, exemplified by the characterisation of~\cite{pic}.

There also remain interesting questions outside the context of causal decompositions and circuit connectivity.
A natural problem is to characterise the set of finite circuits up to syntactical equivalence for given finite sets $\labelsA$, $\labelsB$ and to study its order- and category-theoretic structure.
A beginning has been made in~\cite{SalzS24}, which characterises this set in cases where $\labelsA$ and $\labelsB$ have small numbers of elements.
Also the relation between syntactical inequivalence of circuits and their comparative expressivity in specific SMCs such as classical and quantum theory should be investigated in more detail.

\paragraph{Acknowledgements}
I am grateful to Robin Lorenz, Aleks Kissinger, Sean Tull, Gabriele Tedeschi, Matthias Salzger, Elie Wolfe, Bregt van der Lugt, and the organisers and participants of \emph{Causalworlds 2024} for useful discussions regarding this work and comments on early versions.

    \printbibliography
\end{document}